\documentclass[11pt]{article}

\newcommand{\OutputPart}{all}

\usepackage[utf8]{inputenc}
\usepackage[T1]{fontenc}
\usepackage{amsmath, amssymb, amsthm}
\usepackage{braket}
\usepackage{graphicx}
\usepackage{float}
\usepackage{subcaption}
\usepackage{booktabs}
\usepackage{algorithm}
\usepackage{algpseudocode}
\usepackage[numbers,sort&compress]{natbib}
\usepackage[dvipsnames]{xcolor}
\usepackage{tikz}
\usetikzlibrary{positioning, arrows.meta, fit, backgrounds}
\definecolor{qrDark}{HTML}{2E2A6B}
\definecolor{qrMid}{HTML}{3B3F8F}
\definecolor{qrLeaf}{HTML}{EEF0FA}
\usepackage[colorlinks=true,
  linkcolor=blue,
  citecolor=ForestGreen,
  urlcolor=blue,
  hypertexnames=false]{hyperref}
\usepackage{cleveref}

\makeatletter
\let\NAT@orig@def@last@yr\def@NAT@last@yr
\def\def@NAT@last@yr#1{\NAT@orig@def@last@yr{--\NAT@penalty}}
\makeatother

\crefname{figure}{Figure}{Figures}
\crefname{section}{Section}{Sections}
\crefname{appendix}{Appendix}{Appendices}
\crefname{equation}{Eq.}{Eqs.}
\Crefname{equation}{Equation}{Equations}
\crefname{table}{Table}{Tables}
\crefname{algorithm}{Algorithm}{Algorithms}
\crefname{theorem}{Theorem}{Theorems}
\crefname{lemma}{Lemma}{Lemmas}
\crefname{corollary}{Corollary}{Corollaries}
\crefname{proposition}{Proposition}{Propositions}
\crefname{definition}{Definition}{Definitions}
\crefname{remark}{Remark}{Remarks}

\makeatletter
\edef\@tempa{\noexpand\in@{,\OutputPart,}{,all,paper,supp,}}\@tempa
\ifin@\else
  \@latex@error{Unknown \string\OutputPart\space `\OutputPart'}%
    {Set \string\OutputPart\space to all, paper or supp.}
\fi
\newif\if@droppart
\AddToHook{shipout/before}{\if@droppart\DiscardShipoutBox\fi}
\let\@deferred@protected@write\protected@write
\long\def\@immediate@protected@write#1#2#3{%
  \begingroup
    #2%
    \let\protect\@unexpandable@protect
    \edef\reserved@a{\immediate\write#1{#3}}%
    \let\protect\noexpand
    \reserved@a
  \endgroup
  \if@nobreak\ifvmode\nobreak\fi\fi}
\newcommand{\BeginPart}[1]{%
  \clearpage
  \edef\@tempa{\noexpand\in@{,\OutputPart,}{,all,#1,}}\@tempa
  \ifin@
    \global\@droppartfalse
    \global\let\protected@write\@deferred@protected@write
    \hypersetup{bookmarksdepth}%
    \edef\@tempa{\OutputPart}\def\@tempb{all}%
    \ifx\@tempa\@tempb\else \setcounter{page}{1}\fi
  \else
    \global\@dropparttrue
    \global\let\protected@write\@immediate@protected@write
    \hypersetup{bookmarksdepth=-10}%
  \fi}
\makeatother

\makeatletter
\newif\if@splitbib        
\newif\if@suppbibentries  
\newif\if@suppbib         
\newif\if@bibskipping     
\def\citation#1{\@for\@citeb:=#1\do{\global\@namedef{paper@cited@\@citeb}{}}}
\def\bibdata#1{%
  \global\@splitbibtrue
  \gdef\citation##1{\@for\@citeb:=##1\do{%
    \@ifundefined{paper@cited@\@citeb}{\global\@suppbibentriestrue}{}}}}
\let\@natbib@lbibitem\@lbibitem
\def\@lbibitem[#1]#2{%
  \@endbibskip
  \@tempswatrue
  \if@splitbib
    \@ifundefined{paper@cited@#2}{\@tempswafalse}{}%
    \if@suppbib
      \if@tempswa \@tempswafalse \else \@tempswatrue \fi
    \fi
  \fi
  \if@tempswa
    \expandafter\@natbib@lbibitem
  \else
    \expandafter\@startbibskip
  \fi
  [{#1}]{#2}}
\def\@startbibskip[#1]#2{%
  \global\advance\c@NAT@ctr\@ne
  \global\@bibskippingtrue
  \setbox\z@\vbox\bgroup}
\def\@endbibskip{\if@bibskipping \global\@bibskippingfalse \egroup \fi}
\AddToHook{env/thebibliography/end}{\@endbibskip}
\newcommand{\SupplementalReferences}{%
  \if@suppbibentries
    \begingroup
      \@suppbibtrue
      \renewcommand{\refname}{Supplemental References}%
      \@input@{\jobname.bbl}%
    \endgroup
  \fi}
\makeatother

\theoremstyle{definition}
\newtheorem{definition}{Definition}
\theoremstyle{plain}
\newtheorem{lemma}{Lemma}
\newtheorem{theorem}{Theorem}
\newtheorem{proposition}{Proposition}
\newtheorem{corollary}{Corollary}
\theoremstyle{remark}
\newtheorem{remark}{Remark}

\newcommand{\Tree}{\mathcal{T}}   
\newcommand{\edges}{E}            
\newcommand{\bigO}{O}             
\newcommand{\Fswap}{F_{\mathrm{swap}}}
\newcommand{\Fpur}{F_{\mathrm{pur}}}
\newcommand{\Fth}{F_{\mathrm{th}}}
\newcommand{\fin}{f_{\mathrm{in}}}

\makeatletter
\let\standardfnsymbol\@fnsymbol
\renewcommand{\@fnsymbol}[1]{%
  \ifcase#1\or\textdagger\or\textasteriskcentered\else\standardfnsymbol{#1}\fi}
\makeatother

\title{One-shot Routing in Quantum Networks}

\author{
  Nadav Lavi\thanks{These authors contributed equally.}\hspace{0.3em}\thanks{Viterbi Faculty of Electrical and Computer Engineering, Technion -- Israel Institute of
  Technology, Haifa 32000, Israel.} \and
  Nir Gutman\footnotemark[1]\hspace{0.3em}\footnotemark[2] \and
  Ido Kaminer\footnotemark[2] \and
  Ariel Orda\footnotemark[2]
}

\date{}

\begin{document}

\BeginPart{paper}
\maketitle


\begin{quote}
\small
\noindent\textbf{ABSTRACT}\quad
Distributed quantum computation requires many entangled pairs to be available
simultaneously, so routing must optimize fidelity from a fixed, short-lived set
of network resources rather than the rate of pairs accumulated over time. We
formulate this one-shot routing problem for heterogeneous Werner-state
links and jointly optimize path selection and the schedule of entanglement
swapping and distillation. We show that the common purify-then-swap ordering is
not optimal network-wide and that exact schedule enumeration requires
$m^{\Theta(m)}$ time.
We introduce Fidelity-Optimizing Local Detours (FOLD), a
heuristic that starts from the highest-fidelity path and merges a detour only
when distillation improves on both routes and raises end-to-end fidelity.
FOLD runs in polynomial time, recovers almost all of the optimal fidelity on
tractable instances, and is four orders of magnitude faster than exact
optimization at circuit rank $5$.
On the Internet topologies where any strategy improves on shortest-path
routing, FOLD achieves the highest mean end-to-end fidelity among the
polynomial baselines.
We further extend the formulation to links whose fidelities are known only as
distributions, where uncertainty, rather than the effort spent routing, limits
the fidelity delivered.
Through an example of requests competing for the same links, we show that
coordination becomes a condition for feasibility, and that a shared path budget
controls the fidelity trade-off between them.
\end{quote}

\section{Introduction}
\label{sec:introduction}

Quantum computers have advanced tremendously over the last few years
\citep{arute2019supremacy, acharya2025surfacecode}, and promise to outperform classical
computers on problems such as computational search \citep{grover1996search} and the
simulation of physical systems \citep{feynman1982simulating, lloyd1996universal}. Realizing
this potential at scale requires connecting multiple such computers into a
quantum network \citep{wehner2018vision}, a first step toward the broader vision
of a quantum internet \citep{kimble2008internet}. Connecting nodes over such a
network enables tasks
with no classical counterpart, among them distributed quantum computation \citep{cirac1999distributed,
caleffi2024distributed}, distributed leader election
\citep{tani2012leader} and secure multi-party computation \citep{crepeau2002smpc}.
A full protocol stack for quantum networks now exists \citep{kozlowski2020designing}, and an operating system runs platform-independent applications on real network nodes \citep{delledonne2025operating}.

Quantum networks rely on distributing entanglement across lossy, noisy links, where repeater-based
protocols use local entangled pairs, entanglement swapping \citep{zukowski1993swapping}, and distillation
\citep{bennett1996purification, deutsch1996privacy} to create long-distance high-fidelity Bell pairs for end users (\cref{fig:two_regimes}(a)). Teleportation then consumes those pairs to transfer data \citep{bennett1993teleporting}, and their fidelity sets its reliability.
Making these operations work together at network scale is an active field of research \citep{vanmeter2009, shi2020concurrent, li2021routing, li2022fidelity, wang2023scheduling, liu2025joint}, and the routing layer decides which pairs to combine and in what order.

\begin{figure}[t]
  \centering
  \includegraphics[width=\columnwidth]{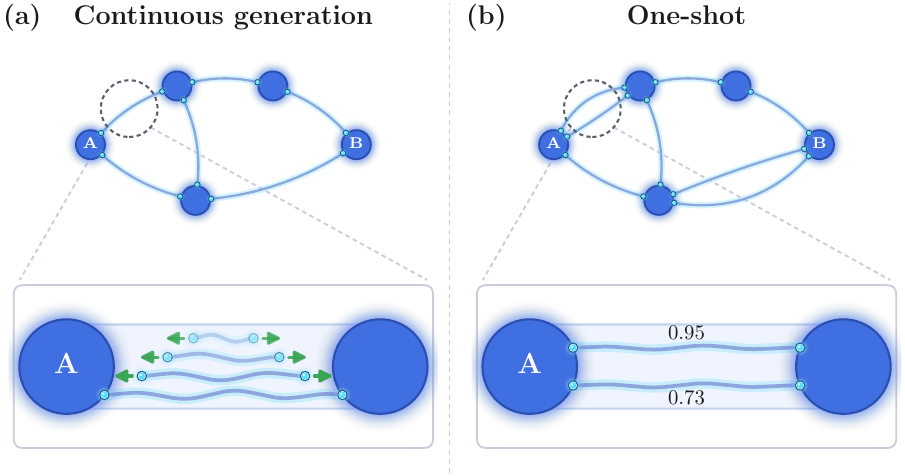}
  \caption{\textbf{Two resource regimes of a quantum network, and the one we
  solve.} Each panel draws the same five-node network, with one of its links
  magnified below. \textbf{(a)} Under continuous entanglement generation, a
  link is a repeater channel that keeps distributing pairs. A pair separates
  inside the channel and one half travels to each node. The pair at the bottom
  has already arrived. Fresh pairs keep coming, so one is as good as another
  and routing maximizes the rate delivered. The channel is the resource, which
  is why a link is drawn as a single edge. \textbf{(b)} One-shot operation is
  the regime we study. The links hold whatever they have
  produced by now, and nothing follows. Nothing travels in the magnified link,
  which holds two pairs that arrived earlier, of fidelity $0.95$ and $0.73$.
  Each pair now carries a value of its own. The snapshot holds eight pairs in
  all, two of the links carrying two apiece, and routing maximizes the
  end-to-end fidelity they reach together. They admit $6{,}862$ sequences of
  swapping and distillation, whose best reaches $F(A,B)=0.821$ by interleaving
  the two operations, ahead of the best purify-then-swap ordering at $0.812$.}
  \label{fig:two_regimes}
\end{figure}

Routing over such a network cannot borrow directly from its classical counterpart. Classical link-state protocols \citep{mcquillan1980, moy1998ospf, fortz2000ospf} rest on copying data, sharing a link, and letting memory hold data until a link frees up. All three fail for quantum networks, which cannot copy entanglement \citep{wootters1982nocloning} or resend it once lost. Every state a route consumes must therefore be available simultaneously.

The literature on entanglement routing focuses mainly on the \emph{rate} of pairs whose fidelity exceeds a threshold $\Fth$, measured by throughput \citep{shi2020concurrent, li2021routing, li2022fidelity, tan2026distributed}, time \citep{goodenough2021, wang2023scheduling,
inesta2023optimal} or resource cost \citep{liu2025joint,
halder2024optimal, jia2024routing}. This objective comes from quantum key distribution \citep{bennett1984qkd}, where it is natural: keys accumulate pair by pair over time, and privacy amplification \citep{deutsch1996privacy} averages away the damage done by imperfect ones. However, distributed quantum computation requires the opposite. Teleporting a multi-qubit register \citep{bennett1993teleporting} or sharing an entangled state among distant processors requires many pairs held \emph{simultaneously}, with no analog of privacy amplification, and a lost or degraded pair cannot be replaced. Single-time-frame fidelity, not rate, bounds the scale of the distributed computation. Such an application therefore needs the fidelity attainable across the pairs it consumes at once -- the quantity a rate objective constrains at $\Fth$ rather than maximizes.

This calls for a different view of the network.
Since entanglement generation and swapping succeed only probabilistically, the pairs available at any moment form a short-lived subset of the network's links, known to the nodes and gone once their memories decohere.
The object is therefore a finite set of pairs, fixed in size and fidelity, to be combined by distillation and swapping within the coherence time, not a rate accumulated over many.

Here we define and study entanglement routing in a \emph{one-shot} network.
Instead of assuming a constant rate of states replenished by quantum repeaters, we fix the set of pre-distributed states the network holds at a given moment
(\cref{fig:two_regimes}(b)) and
ask: given the entangled pairs available now, which paths and which schedule of distillation and swapping
attain the highest end-to-end fidelity between two users? We optimize end-to-end fidelity because quantum teleportation carries the messages.
We show that the purify-then-swap ordering assumed by much of the literature is not optimal network-wide, and that exact optimization does not scale: our optimal solver enumerates every valid schedule and
runs in $m^{\Theta(m)}$ worst-case time, where $m$ is the number of edges. We analyze optimal solutions across many topologies, and \cref{sec:routing} uses these findings to design an algorithm that outperforms existing approaches. The same model accommodates uncertain edge fidelities (\cref{sec:uncertainty}) and several simultaneous requests (\cref{sec:requests}).

Most work on quantum networks addresses single-path routing \citep{shi2020concurrent,li2021routing,halder2024optimal} or the scheduling of operations within a given path \citep{goodenough2021, fan2025distribution, vanmeter2009, jia2024routing, liu2025joint, wang2023scheduling}. To our knowledge, no existing method chooses several routes and the ordering of swapping and distillation together. Our proposed algorithm FOLD does both at once. On any graph it selects a multi-path and the schedule that combines it, so the ordering of operations follows the graph rather than being fixed in advance.

Building on this analysis, we establish FOLD (Fidelity-Optimizing Local Detours), a scalable heuristic for the same problem (\cref{sec:heuristic_algorithm}). FOLD runs in $\bigO\!\big(r\,(n^{2}+m)+n\,m\log n\big)$ time
(\cref{sec:heuristic_complexity}),
where $n$ is the number of nodes, $m$ the number of edges, and $r=m-n+c$ the circuit rank of
the graph, since each merge it performs consumes one independent cycle. It
recovers almost all of the fidelity the optimal algorithm attains while running four orders of
magnitude faster at circuit rank $5$, where the exact solver already takes about an hour per
instance. We benchmark FOLD on the Internet topologies of the Topology
Zoo~\citep{knight2011zoo} (\cref{sec:benchmark}), where it doubles the fidelity gain of the strongest
reference. It does not always reach the optimum. On \emph{reducible} topologies
it recovers about four fifths of the optimal fidelity gain, slightly behind a
recursive algorithm that consumes such a graph entirely~\citep{meng2023series}.
Most deployed topologies are not reducible, so this case is rare in practice.
\section{One-shot Network Model and Problem Statement}
\label{sec:model}

Our model maximizes the fidelity between two users of the network: given a graph, which sequence of swapping and distillation attains the best end-to-end connection quality, and how do we find that sequence efficiently?
In quantum networks, \emph{swapping} and \emph{distillation} establish end-to-end entanglement (\cref{fig:operations}(a)). Different schedules of these operations may establish the same end-to-end connection with different fidelities (\cref{fig:operations}(b)).
The challenge is that the space of feasible routing configurations becomes intractable even for small graphs.

A one-shot quantum network is a graph $G=(V,E)$, with $n=|V|$ nodes and $m=|E|$ edges, where each edge $(u,v) \in E$ is a 2-qubit entangled state with one qubit held at each of its endpoints (\cref{fig:two_regimes}(b)).
Studies have assumed a constant flow of entangled states along each edge, continuously replenished by quantum repeaters.
Instead, in our study each link holds a fixed, finite number of entangled states, which raises a new algorithmic routing challenge.
We use the term \emph{edge} for the entangled state between two nodes, and the term \emph{link} for several parallel edges between the same two nodes (\cref{fig:operations}(b)-(c)).
We call the two nodes that wish to communicate the \emph{source} and the \emph{target}, and reserve \emph{endpoints} for the two nodes joined by a given edge.

Quantum teleportation carries a data qubit across the network by consuming a shared entangled pair \citep{bennett1993teleporting}, and it does so faithfully only from maximally entangled states. A pair distributed through a noisy channel is instead mixed.

\begin{figure}[H]
  \centering
  \includegraphics[width=\columnwidth]{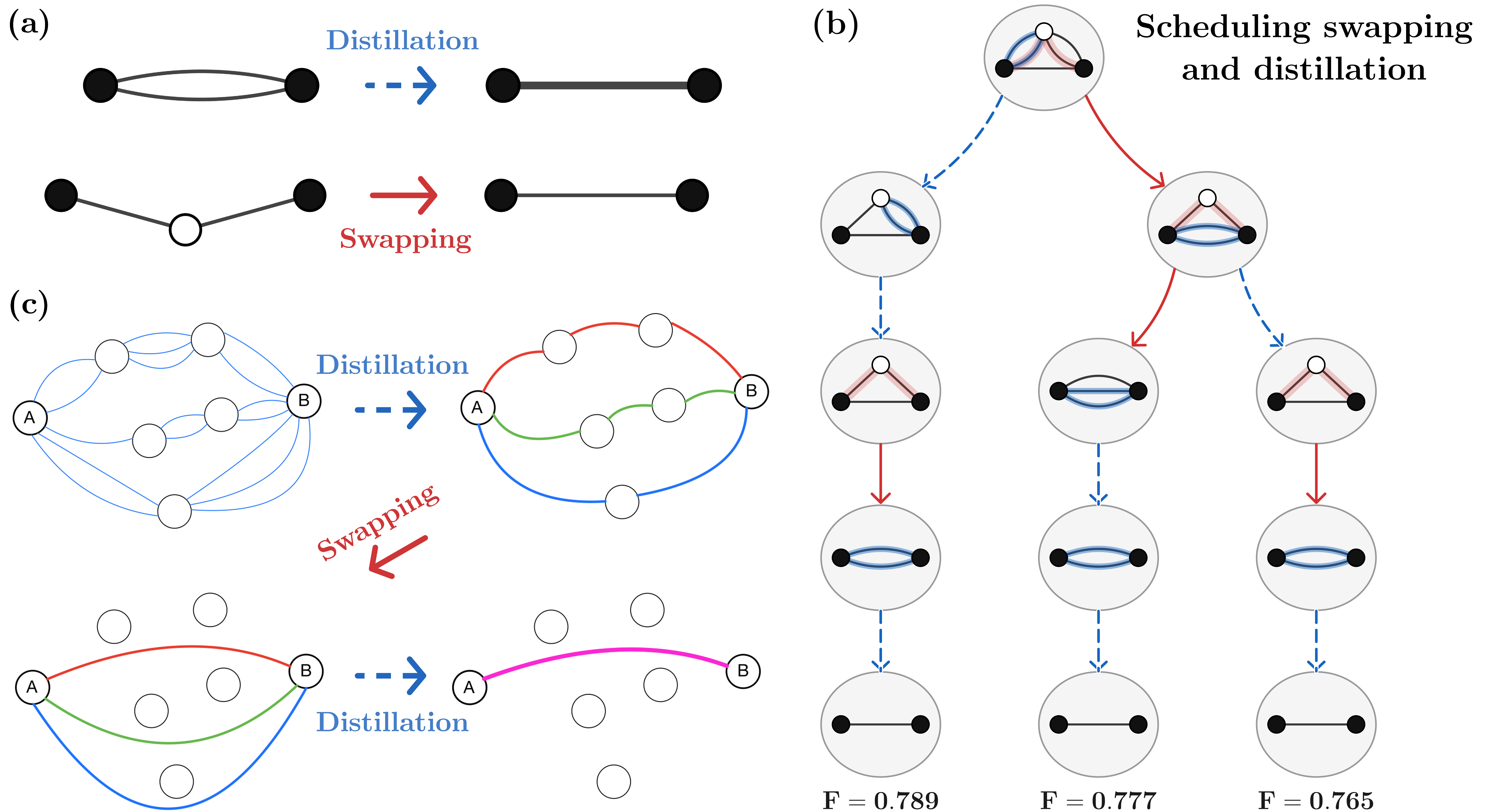}
  \caption{\textbf{Scheduling of swapping and distillation operations in
  quantum networks.} \textbf{(a)} \emph{Distillation} consumes two parallel
  edges between the same pair of nodes to produce a single higher-fidelity
  edge, while \emph{swapping} consumes two edges meeting at an intermediate
  node to produce a direct edge between their outer endpoints. \textbf{(b)} Two
  users share a direct edge of fidelity $0.745$ and a two-hop route through a
  repeater, each hop a link of two edges of fidelity $0.83$. Every branch
  reduces this graph to a single end-to-end connection through four operations,
  swapping (red, solid) and distillation (blue, dashed), and every branch
  spends all five edges. The schedules still deliver different fidelities.
  Distilling inside each link before swapping reaches $F=0.789$. Swapping both
  routes first and merging them afterwards reaches $F=0.777$. The right branch
  performs the same operations as the middle one and merges the direct edge
  earlier, which costs a further $0.012$. Fidelity therefore depends on which
  states an operation combines, not on how many operations a schedule runs. We
  enumerate this space exactly, and propose FOLD, a polynomial-time heuristic
  that outperforms existing scheduling approaches (\cref{sec:routing}).
  \textbf{(c)} A larger example: distilling the parallel edges of each link,
  swapping the resulting links along the paths, and finally distilling the two
  $A$--$B$ edges together collapses the whole graph into one high-fidelity
  end-to-end connection.}
  \label{fig:operations}
\end{figure}

We measure the quality of an entangled state by its \emph{fidelity} to the singlet Bell state $\ket{\Psi^-}$,
\begin{equation}
  F = \bra{\Psi^-} \rho \ket{\Psi^-},
  \label{eq:fidelity_def}
\end{equation}
the overlap between the shared, generally noisy, two-qubit state $\rho$ and this ideal target, which takes values $0\le F\le 1$.
Fidelity is therefore the natural objective for a quantum network: the closer $F$ is to $1$, the more faithfully a pair teleports a qubit.
Each edge in our model is a Werner state, a mixed state of the form:
\begin{equation}
    W_F = F \ket{\Psi^-}\bra{\Psi^-} + \frac{1-F}{3}(\ket{\Psi^+}\bra{\Psi^+} + \ket{\Phi^+}\bra{\Phi^+} + \ket{\Phi^-}\bra{\Phi^-}),
    \label{eq:first_werner_form}
\end{equation}
where $F$ is its fidelity as defined in \cref{eq:fidelity_def}. For a Werner state, $F=1$ recovers the perfect $\ket{\Psi^-}$ Bell state and $F=\tfrac14$ the totally mixed two-qubit state. The state is entangled, and thus useful for communication, only for $F>\tfrac12$.

Werner states model noisy entanglement in quantum networks: they capture the fraction of $\ket{\Psi^-}$ in the total state, and they are the output of a depolarizing channel, the standard noise model.

We assign each edge $e \in E$ its own such state, with fidelity $f_e = \bra{\Psi^-} W_{f_e} \ket{\Psi^-}$.
These fidelities are \emph{heterogeneous}: distinct edges may carry arbitrarily different values of $f_e$, reflecting the different lengths, hardware and noise levels of the physical links that produced them, as well as the time each state has already spent in memory.
An instance of our problem is therefore a graph $G=(V,E)$ together with a fidelity assignment $\{f_e\}_{e \in E}$, and nothing in our analysis or algorithms assumes a common initial fidelity. Where we do fix a uniform value, it is only to isolate a single effect in a controlled experiment.
Throughout the paper we write $f_e$ for the fidelity of an individual edge and reserve the capital $F$ for the fidelity of a state produced by a protocol, such as the end-to-end fidelity $F(P)$ delivered along a path $P$.

\emph{Swapping} entangles two nodes by consuming two adjacent states that share
an intermediate node (\cref{fig:operations}(a)) \citep{zukowski1993swapping}.
\emph{Distillation}, often also called \emph{purification}, instead consumes
several noisy pairs and returns fewer pairs of higher fidelity
\citep{bennett1996purification}.
On Werner states each operation maps the fidelities $f_1,f_2$ of the two states
it consumes onto a single output fidelity,
\begin{align}
  \Fswap(f_1,f_2) &= f_1f_2 + \frac{(1-f_1)(1-f_2)}{3},
  \label{eq:swapping_fidelity}\\[2pt]
  \Fpur(f_1,f_2) &= \frac{f_1f_2+\frac{1}{9}(1-f_1)(1-f_2)}
       {f_1f_2+\frac{1}{3}(f_1+f_2-2f_1f_2)+\frac{5}{9}(1-f_1)(1-f_2)}.
  \label{eq:distillation_fidelity}
\end{align}
We adopt the BBPSSW relation of \cref{eq:distillation_fidelity}
\citep{bennett1996purification} throughout. This relation gives the fidelity of
the retained pair conditioned on the favorable BBPSSW measurement outcome.
Throughout the main analysis, we optimize end-to-end fidelity conditioned on
successful completion of the selected schedule.

Distillation raises \emph{both} inputs only when their fidelities are close
enough, so $\Fpur$ carves out a bounded \emph{purifiable region} that makes
distillation a resource to spend selectively rather than a step to apply
wherever the topology allows (\cref{fig:distillation_selectivity}(a),
analyzed in \cref{sec:purifiable_region}).
Both relations accept arbitrary inputs -- elementary edges, or states that
earlier operations produced -- so they compose recursively, and capital
arguments such as $\Fswap(F_1,F_2)$ mark an input of the second kind.
Sequencing swaps along a path $P$ collapses its edges into a single end-to-end
fidelity,
\begin{equation}
    F(P)=\frac{1}{4}\Big( 1+3 \prod_{e\in P} \frac{4f_e-1}{3} \Big),
    \label{eq:path_fidelity}
\end{equation}
with $p_e=\tfrac{4f_e-1}{3}$ the depolarizing parameter of edge $e$
(\cref{sec:operations} derives the three relations). The companion visualizer applies both operations interactively and traces the purifiable region~\citep{entanglenet_visualizer}.

For most of this work, we set aside the probability that the selected schedule
completes successfully and focus on its conditional output fidelity. For
applications in which achieving high state fidelity is the primary requirement,
this objective can be relevant even at the cost of a lower success probability.
Except for the distillation-order comparison in \cref{sec:distillation_order},
our results do not account for delivery probability, expected utility, or
entanglement-generation rate.
The figures in the rest of this paper omit the individual qubits and show only the edges that connect them.
\section{Routing}
\label{sec:routing}
\label{sec:optimal_brute_force_algorithm}

Routing decisions in a quantum network differ from their classical counterparts for the reasons given in \cref{sec:introduction}. Each graph edge supports one operation rather than being shared among many. No-cloning~\citep{wootters1982nocloning} prevents copying a data qubit for a later transmission, so all states required by a route must be available simultaneously.
Quantum routing therefore consists of two major parts: \emph{path selection} and \emph{scheduling}.
In \emph{path selection} the algorithm selects the path, or paths, connecting the source and target nodes, whereas in \emph{scheduling} it chooses the order of distillation and swapping along these paths; any operation on more than two states also requires its own granular ordering on the states.
This yields many possible strategies given any graph.
In this section we address all parts of routing in one-shot quantum networks and introduce FOLD, a new heuristic algorithm.
We employ the technique of multi-path routing in our algorithm to take advantage of all the available resources of a one-shot quantum network, as seen for example in \cref{fig:operations}(c).

To experiement with the optimal solutions, we also built an exact solver that composes every valid sequence of swapping and distillation and returns the best of them (\cref{lem:exhaust}). \cref{fig:strategy_histogram} shows what such an enumeration contains: every one of the possible schedules, each scored by the fidelity it delivers.
It also compares how different methods perform on this example graph.
The exhaustive algorithm runs at the price of an $m^{\Theta(m)}$ running time (\cref{thm:brute}). This confines it to small graphs -- even at circuit rank $5$, which nine edges already reach, a single instance takes about an hour to exhaustively search over all possible schedules.
This section reports the experiments we ran with it. Each asks one question about optimal solutions, and each answer shapes FOLD, the polynomial-time heuristic assembled in \cref{sec:heuristic_algorithm}. The implementation of both algorithms is available in the paper's GitHub repository~\citep{entanglenet_code}.

\begin{figure}[H]
  \centering
  \includegraphics[width=0.9\columnwidth]{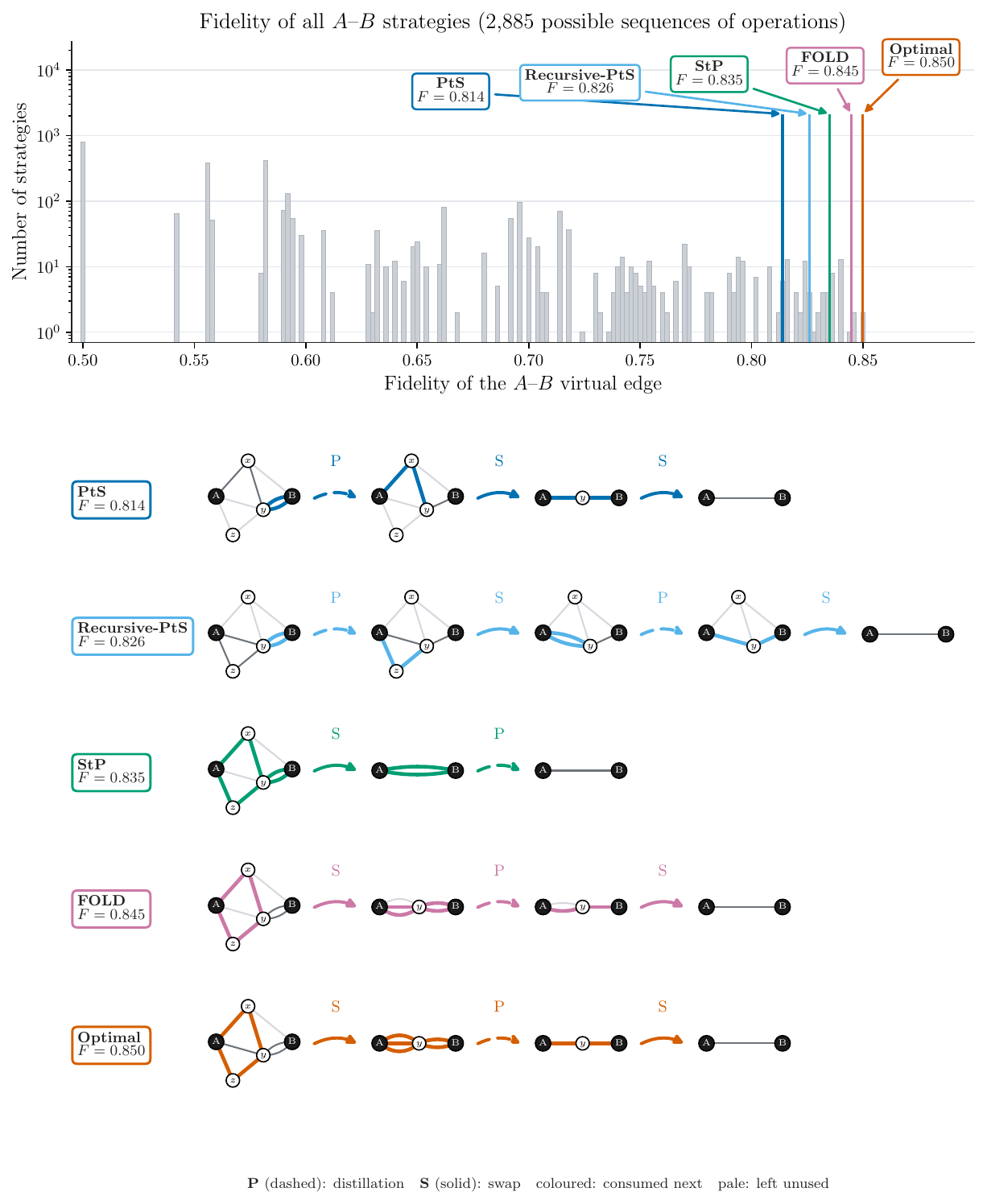}
  \caption{\textbf{Every schedule connecting two nodes, enumerated: the optimum
  interleaves swapping and distillation rather than following either canonical
  ordering, and FOLD nearly reaches it.}
  Fidelity of all $2{,}885$ operation sequences connecting $A$ and $B$ on the
  graph shown. Each marked strategy below is drawn as purify (P) and swap (S)
  operations, the edges an operation consumes colored and any unspent edge left
  pale. The strategies are purify-then-swap (PtS), swap-then-purify (StP), the
  recursive purify-and-swap reduction (\emph{Recursive},
  \cref{sec:benchmark})~\citep{meng2023series} and \emph{FOLD}, which leads all
  three. PtS, the ordering most schedulers assume, is the weakest here.}
  \label{fig:strategy_histogram}
\end{figure}

\subsection{Path Selection}
\label{sec:path_selection}

Routers in classical networks select a path that minimizes latency or the number of links used, both additive objectives, so a path costs the sum of its edges. 
Unlike these additive objectives, we maximize the end-to-end fidelity that swapping creates along a path. The edge fidelities combine multiplicatively along that path, so weighting each edge appropriately reduces the maximization to a shortest-path problem, as \citet{difranco2012optimal} do.
Specifically, rearranging \cref{eq:path_fidelity} and weighting each edge $e$ by $w_e=-\log(4f_e-1)+\log 3$ 
turns a product over the edges of a path into a sum,
\[
  \sum_{e\in P} w_e = \sum -\log(4f_e-1)+\log 3 = -\log \big( \prod_{e\in P}\frac{4f_e-1}{3} \big) =-\log\frac{4F(P)-1}{3}.
\]
This decreases the last term as $F(P)$ rises, so the lightest, i.e. shortest, path is the one of highest end-to-end fidelity.
This formulation brings us back to the realms of a shortest-path problem, where we can use Dijkstra's algorithm~\citep{dijkstra1959}.

This resembles Internet traffic engineering, which derives link weights that reduce its objective to shortest-path routing~\citep{fortz2000ospf,fortz2002traffic}. Protocols flood the weights so that every router runs Dijkstra locally~\citep{moy1998ospf}. 
Throughout this paper, \emph{shortest path} denotes a path whose fidelity is highest.

This reduction requires justification. \citet{difranco2012optimal} showed that for general mixed states, Dijkstra's algorithm cannot find the maximum-fidelity route: path fidelity may depend on two competing quantities, so extending a path can reverse which route is better, violating Bellman's optimality principle. The obstruction is asymmetry rather than mixedness. 
In contrast, the symmetry of the Werner states used in our model reduces the path fidelity (\cref{eq:path_fidelity}) to a single product. 
Werner states naturally describe Bell pairs affected by depolarizing noise: sending one half of $\ket{\Psi^-}$ through a depolarizing channel produces a Werner state~\citep{horodecki1999general}.
Such channels compose multiplicatively, and swapping two Werner links of depolarizing parameters $p_1$ and $p_2$ (\cref{eq:path_fidelity}) returns a Werner link of parameter $p=p_1p_2$.
The metric is closed under swapping, the operation routing performs. \citet{difranco2012optimal} identify pure states as one exception to the obstruction. The Werner family is a second exception.

Classical multipath routing judges a path set by \emph{independence}, \emph{quantity} and \emph{uniformity}, preferring comparable path costs to high variance~\citep{lee2002survey}. All three reappear with different emphasis. Independence becomes a hard constraint, each entangled edge being consumed by a single operation. The purifiable region imposes uniformity too, admitting a distillation only between close fidelities. And where classical multipath \emph{splits} traffic to add capacities, distillation here \emph{recombines} the paths into one connection of higher fidelity. Hence the question: since uniformity is what makes a merge possible at all, should the algorithm seek several similar routes, or one better route?

\cref{fig:path_selection} answers it in three settings of increasing size, and all three reach the same conclusion. \cref{fig:path_selection}(a) shows the simplest case: three parallel edges join the same pair of nodes, each already swapped down into a candidate route -- a better route $e_1$ of fidelity $f_1$ and two routes of equal or lower fidelity.
Sweeping the values $(f_2,f_3)$ below $f_1$, the optimum is always $e_1$ alone or all three distilled together, \emph{never} the weaker pair on its own. This holds even along $f_2=f_3$, where two \emph{equal} weaker routes overtake the single better one above a crossover fidelity of $0.67$ but still lose to distilling all three.
The triangle of \citet{difranco2012optimal} splits into three regions rather than two: direct path $P_1$, detour $P_2$, and a combined region where $P_1$ and $P_2$ are close enough in fidelity to be worth distilling (\cref{fig:path_selection}(b)). In the bipartite topology, where at most two of four candidate paths may be used at once, the optimum shifts steadily towards a single path as the spread of edge fidelities widens (\cref{fig:path_selection}(c)). Similarity makes two routes \emph{mergeable}, but it is no reason to \emph{choose} them, and in these examples the best route almost always belongs to the optimum.

\begin{figure}[H]
  \centering
  \includegraphics[width=\columnwidth]{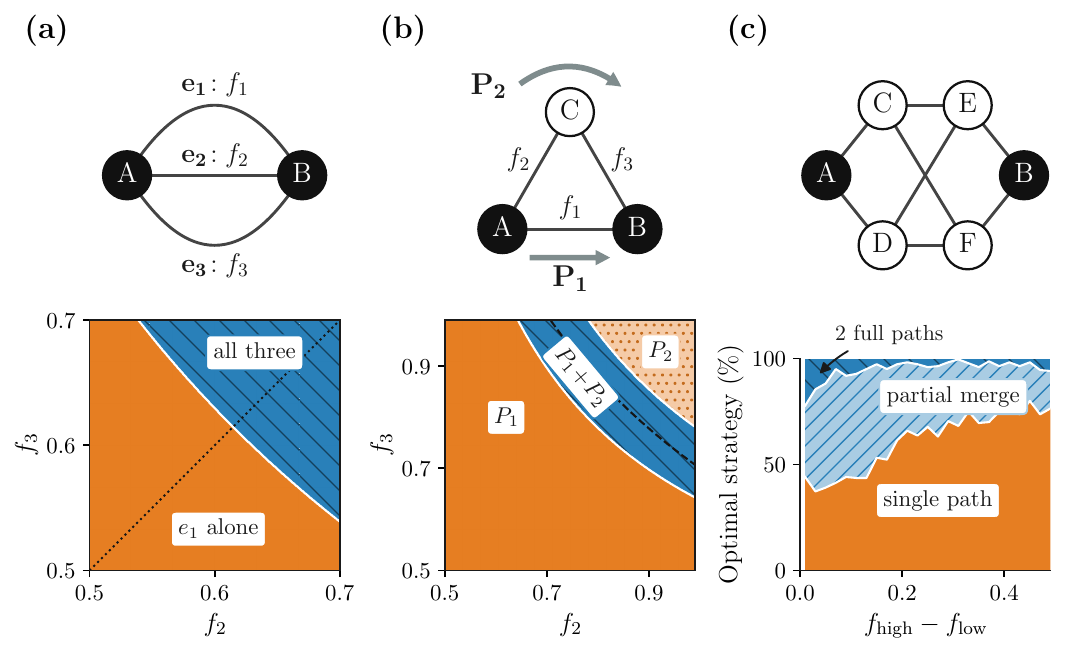}
  \caption{\textbf{Optimal solutions include the highest-fidelity route, and
  additional routes help only when their fidelities are close.} Across the three panels a
  warm hue marks an optimum that uses a single route on its own and blue one
  that combines routes by distillation.
  Each topology is sketched above the panel it belongs to.
  \textbf{(a)} Three parallel edges standing for candidate routes between the
  same pair of nodes -- a better route $e_1$ of fidelity $f_1=0.7$ and two
  weaker ones $e_2,e_3$. Over $(f_2,f_3)$ the optimum is either $e_1$ alone or
  all three distilled together, so distilling only the two weaker routes claims
  no region of the plane at all. The dotted line marks $f_2=f_3$.
  \textbf{(b)} The triangle of \citet{difranco2012optimal}, with the direct
  path $P_1$ held at $f_1=0.7$ and the detour $P_2$ obtained by swapping the
  edges of fidelity $f_2$ and $f_3$. The plane splits into $P_1$, $P_2$ and a
  combined band lying between them, the dashed line marking where $P_1$ and
  $P_2$ deliver equal fidelity. \textbf{(c)} The bipartite topology, in which
  Alice and Bob are joined by four candidate paths of which at most two may be
  used at once, either disjoint (2 full paths) or sharing an intermediate node
  (a partial merge). Over $5{,}500$ optimal solutions, $220$ per bin, with
  every edge fidelity drawn uniformly from a window of width
  $f_{\mathrm{high}}-f_{\mathrm{low}}$, the single-path share grows from $44\%$ to $76\%$ as that
  window widens.}
  \label{fig:path_selection}
\end{figure}

The reduction to a shortest-path problem assumed a route draws on a single chain of edges, each spent once, its fidelity depending only on that chain. Distillation breaks this assumption: \Cref{eq:distillation_fidelity} is neither additive nor multiplicative, so once distillation is allowed a route has no fidelity of its own: its quality depends on how many further edges are spent improving it, and every edge may be spent once. Committing an edge to one segment denies it to every other, and the best sub-path in isolation need not belong to the best overall solution. Optimal substructure is lost, and with it any guarantee that a shortest-path computation returns the best connection. \citet{fan2025distribution} recover optimal substructure by fixing the source--target route in advance, which lets a dynamic program build their purification-augmented swapping trees, but we cannot fix the route: which edges to spend on which segment is itself part of the decision. The candidate set need not be large, though: the exhaustive solver's enumeration shows which paths the optimum actually draws on, and the shortest paths, being those of highest fidelity, usually take part in it (\cref{fig:strategy_histogram}). Restricting the candidate set to the shortest paths is a limit \citet{halder2024optimal} reach analytically, capping the usable length of a path at a fidelity threshold $\Fth$. We impose no such threshold, since fidelity is our objective rather than a bar to clear, yet the same decay constrains us. Specializing \cref{eq:path_fidelity} to a homogeneous path whose $L$ intermediate nodes all carry fidelity $\fin$ gives $F=\tfrac14+\tfrac34\big(\tfrac{4\fin-1}{3}\big)^{L+1}$, so every swap costs a constant factor of fidelity and a long path is not worth the loss.

We therefore take FOLD's candidate set to be the $k_{max}$ shortest paths, computed with Yen's algorithm~\citep{yen1971kshortest} as in~\citet{shi2020concurrent,li2021routing,halder2024optimal}. Yen builds each new path as a \emph{deviation} of one already found, reusing a prefix and rerouting only from some node onward, so successive candidates differ from the running solution by local detours -- exactly the segments the algorithm merges. It thus starts from the shortest path and merges further routes \emph{into} it, rather than searching for a set of mutually similar paths.

Many prior works address path selection alone, taking the operations as given and asking only which path, or paths, carry the entanglement \citep{shi2020concurrent, li2021routing, halder2024optimal}, in centralized \citep{li2022fidelity} and distributed \citep{tan2026distributed} settings. \citet{jia2024routing} combine path selection with scheduling in an FPTAS that jointly routes and schedules purification, yet still return a single path per request.

\subsection{Scheduling}
\label{sec:purifiable_region}

With the routes chosen, what remains is a \emph{scheduling} task: in what order to combine the edges of the chosen paths, and in what order to apply each pairwise operation. Schedules grow super-exponentially with the number of edges, so exact optimization is out of reach at realistic size and any scalable algorithm must be a heuristic. A heuristic is only as good as the structure it exploits, so we first ask what an optimal schedule looks like. Prior work asks a narrower question, fixing a path and optimizing only the \emph{order} of operations along it, either by searching candidate orderings \citep{goodenough2021, fan2025distribution} or by committing to one a priori \citep{vanmeter2009, goodenough2021}. Those commitments often pick one of two canonical orderings. Purify-then-swap (PtS) distills the physical edges first and swaps only afterwards, and swap-then-purify (StP) swaps each route down to a single link and distills the results.

\cref{fig:strategy_histogram} enumerates all $2{,}885$ schedules connecting $A$ and $B$ on the graph drawn there, and on that graph the ordering most schedulers assume is already the weaker of the two. StP reaches $F=0.835$ and leads PtS \citep{jia2024routing, liu2025joint, wang2023scheduling} by $0.021$, a margin a scheduler committed to PtS gives up before it makes any other decision. Neither ordering reaches the optimum, which at $F=0.850$ follows neither and \emph{interleaves} swaps and distillations instead.

The conclusion holds network-wide, not only within one schedule. Many algorithms default to PtS, purifying the links before joining them \citep{jia2024routing, liu2025joint, wang2023scheduling}. This is optimal along a single path but not once a network offers more than one route from $A$ to $B$, which is the common case: across more than $3{,}000$ small random networks the two orderings disagree on nearly a third. Topology alone decides which ordering is better -- parallel copies of a link favour PtS, disjoint routes favour StP -- and the gap grows with every disjoint route added. No canonical ordering applied globally reaches the optimum, so the ordering must follow the topology segment by segment as the solution is built, rather than be fixed once for the whole network.

One partial guarantee exists. \citet{jia2024routing} prove that PtS is optimal whenever every edge satisfies $f_e\ge0.7$ and every swap succeeds with probability $p_s\le0.818$. However, the guarantee holds \emph{along a given path} and does not extend network-wide, as \cref{fig:strategy_histogram} shows. Outside that range the optimal ordering is unproven even along a single path. Their objective also inverts ours: they minimize cost subject to a fidelity constraint, a problem they show to be NP-hard, and \citet{liu2025joint} likewise minimize the states consumed to reach a target fidelity. We fix the available resources and maximize the fidelity they yield.

Optimal schedules often leave applicable distillations unused. As introduced in \cref{sec:model}, the distillation function \cref{eq:distillation_fidelity} improves \emph{both} inputs only inside the \emph{purifiable region}, where $\Fpur(f_1,f_2)>f_1$ and $\Fpur(f_1,f_2)>f_2$ (\cref{fig:distillation_selectivity}(a)). Distilling whenever possible is therefore not optimal. The region admits only inputs of similar fidelity, the largest difference that still improves both being $0.076$, attained at $(f_1,f_2)\approx(0.811,0.735)$. The restriction applies already on the simplest topology. The optimum on two nodes joined by five parallel edges of random fidelity keeps the single best edge and distills nothing in $52.1\%$ of $5{,}000$ draws (\cref{fig:distillation_selectivity}(b)), because no other edge lies close enough in fidelity to purify with it. The draws that do distill merge edges of similar fidelity, a two-edge merge spanning only $0.019$ in fidelity at the median, which mirrors the shape of the region itself.

A whole network is equally selective and distills nothing until its edge fidelities pass a threshold. On the 8-node topology of \cref{fig:distillation_selectivity}(c), every edge carrying the same input fidelity $\fin$, the optimum coincides with the best single path while $\fin$ is low and only above $\fin\approx0.69$ starts purifying and recruiting edges beyond the shortest path. Distillation is a resource to spend selectively, not a step to apply wherever the topology allows.

The region is also known as \emph{banded distillation}, introduced by \citet{vanmeter2009} and used by \citet{goodenough2021} to prune their protocol search, which distills two states only when $|F_1-F_2|\le\varepsilon_{distill}$. Our purifiable region follows directly from the distillation function instead of being tuned, so it needs no threshold $\varepsilon_{distill}$ and never rejects a merge that would have helped. \citet{fayyaz2026purify} likewise derive the $0.076$ bound rather than tune it. The asymmetry there is temporal, with one pair waiting in memory while the other is generated. Ours is spatial, between the parallel routes a network offers from $A$ to $B$. For routing it acts as a gate: the algorithm evaluates a candidate distillation only when its two input fidelities fall inside the region, which costs nothing and discards the merges that cannot help.

\begin{figure}[H]
  \centering
  \includegraphics[width=0.95\columnwidth]{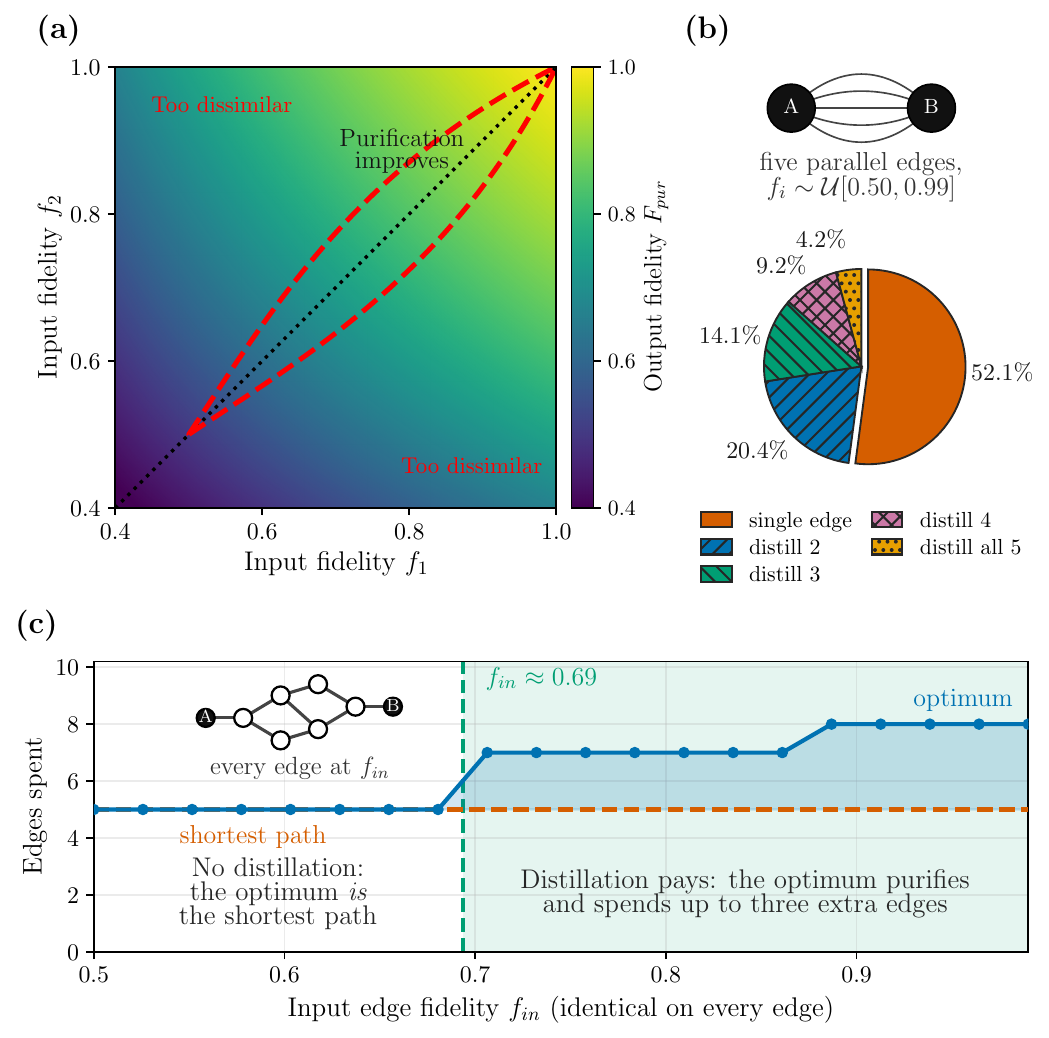}
  \caption{\textbf{Distillation improves only inside the
  \emph{purifiable region}: a single edge is optimal about half the time, and
  the shortest path stays optimal until fidelities cross a threshold.} The same
  restriction, seen at three scales.
  \textbf{(a)} The distillation function: output fidelity $\Fpur(f_1,f_2)$ of
  \cref{eq:distillation_fidelity}, with the red dashed curve bounding the region
  where distillation strictly improves upon \emph{both} inputs and the dotted
  line marking $f_1=f_2$. The region requires both inputs above $f=0.5$ and,
  above that, excludes pairs whose fidelities are too dissimilar.
  \textbf{(b)} A pair of nodes: two nodes joined by five parallel edges of
  independently randomized Werner-state fidelity, over $5{,}000$ draws, with
  the optimal strategy keyed by both color and fill pattern. Keeping the single
  best edge, and distilling nothing, is optimal in $52.1\%$ of cases.
  \textbf{(c)} A network: on an 8-node topology whose edges all carry the same
  input fidelity $\fin$, the $A\!\to\!B$ optimum spends exactly the five edges
  of the shortest path until $\fin\approx0.69$, and only above that threshold
  does it purify and recruit further edges.}
  \label{fig:distillation_selectivity}
\end{figure}

\subsection{The FOLD Algorithm}
\label{sec:heuristic_algorithm}

The experiments above characterize an optimal solution but do not yield one efficiently. The exact solver behind them grows super-exponentially with the size of the network, and more precisely with the number of independent cycles that size brings (\cref{sec:optimal_brute_force_algorithm}). It therefore reaches only small or near-tree topologies, and stalls on the cycle-rich ones, where distillation helps most. These experiments instead expose structure that a tractable algorithm exploits to recover almost all of that optimum. These findings assemble into FOLD, a single polynomial-time heuristic (\cref{alg:path_purification}). It is the first we know of to decide at all three scales at once -- which routes to draw on, in what sequence to combine them, and in what order to merge each pair.

Each of its decisions follows a lesson from the sections above. The best available path consistently belongs to the optimum (\cref{sec:path_selection}), so the algorithm initializes its running solution with the shortest -- highest-fidelity -- path and only ever tries to improve it. The optimal strategy mostly utilizes the high-fidelity paths. Those can be computed using Yen's algorithm, which returns the remaining candidates as local detours, so the algorithm folds them in one at a time. Distillation improves both inputs only inside the purifiable region (\cref{sec:purifiable_region}), so the algorithm merges a detour only when the merge improves both inputs, rather than wherever the topology allows. No canonical swap-distill order wins network-wide (\cref{fig:strategy_histogram}), so the algorithm sets the order merge by merge, as the topology at each divergence dictates. It applies a merge only when the end-to-end fidelity improves.

Concretely, FOLD computes the $k_{max}$ shortest paths and takes the shortest as its running solution $p_1$, then considers the remaining candidates in turn. For a candidate $p_2$, the \emph{diverging sub-paths} $\mathrm{Div}(p_1,p_2)$ are the alternative routes the two take between each pair of consecutive shared nodes: parallel virtual edges between the same anchors, precisely the local detour by which Yen's algorithm produced the candidate. The algorithm skips a merge whose candidate sub-path reuses an already-consumed edge, and one whose input fidelities fall outside the purifiable region. Otherwise it merges the two sub-paths, replaces the old segment of the running solution with the result, and keeps the change only if the end-to-end fidelity improves. It repeats until a candidate yields no further improvement, moves on to the next, and finally collapses the running solution into a single source--target edge.

\begin{algorithm}[htb]
\caption{FOLD: Fidelity-Optimizing Local Detours}
\label{alg:path_purification}
\begin{algorithmic}[1]
\Require Network $G$, source $s$, target $t$, path budget $k_{\max}$
\Ensure Running solution $p_1$, collapsing by swapping to a single $s$--$t$ edge
\Statex
\Function{Fold}{$G,s,t,k_{\max}$}
  \State $P \gets \textsc{Yen}(G,s,t)$ \Comment{$k$ shortest paths, weight $w_e$}
  \State $p_1 \gets P[1]$ \Comment{running solution}
  \For{$i \gets 2$ \textbf{to} $k_{\max}$}
    \Repeat
      \State $\mathit{improved} \gets \textsc{False}$
      \ForAll{$(p_1',p_2') \in \mathrm{Div}(p_1,P[i])$} \Comment{diverging sub-paths}
        \If{$\mathrm{used}(p_2') \cap \mathrm{used}(p_1) \neq \emptyset$}
          \State \textbf{continue} \Comment{edge already consumed}
        \EndIf
        \If{$\big(F(p_1'),F(p_2')\big) \notin$ Purifiable Region}
          \State \textbf{continue}
        \EndIf
        \State $m \gets \Call{Merge}{p_1',p_2'}$ \Comment{distill into one edge}
        \State $p_{\mathrm{new}} \gets p_1$ with sub-path $p_1'$ replaced by $m$
        \If{$F(p_{\mathrm{new}}) > F(p_1)$} \Comment{fidelity improved}
          \State $p_1 \gets p_{\mathrm{new}}$; \quad $\mathit{improved} \gets \textsc{True}$; \quad \textbf{break}
        \EndIf
      \EndFor
    \Until{$\mathit{improved} = \textsc{False}$} \Comment{$P[i]$ exhausted}
  \EndFor
  \State \Return $p_1$
\EndFunction
\end{algorithmic}
\end{algorithm}

Here $F(\cdot)$ is the end-to-end fidelity a route delivers, and $\textsc{Merge}$ swaps each sub-path down to a single edge and distills the two results. Two diverging sub-paths sharing only their anchors are parallel edges, where PtS and StP coincide, nothing being left to swap. With disjoint intermediate nodes there are no parallel copies to distill, so swapping each down first is the only option. Because divergence is split at every shared node, every merge is of one kind or the other, so what the topology fixes is not an ordering but the granularity at which each distillation lands. Merged sub-paths become virtual edges for later merges, so the schedule interleaves swaps and distillations as the optimum of \cref{fig:strategy_histogram} does, unlike the many algorithms optimizing over a fixed path~\cite{vanmeter2009,goodenough2021,koutsopoulos2024,liu2025joint,wang2023scheduling}. These are the same local bypasses Q-PASS holds as recovery detours around a failed segment~\citep{shi2020concurrent}, except that we distill them into the running solution rather than keep them as backups.

Folding candidates one at a time is also the strongest of the three natural distillation orders. Three schedules distill a fixed set of states: \emph{parallel}, \emph{sequential} (entanglement pumping) and \emph{tournament}. Comparing them in the one-shot model, where nothing is replenished, reproduces the ranking \citet{vanmeter2009} establish for repeater chains fed by new pairs (\cref{sec:distillation_order} of the supplemental material). Pumping delivers the highest average fidelity once the probability of success is charged for, and the rule choosing which edges to distill barely matters. An algorithm that maintains a single running solution is structurally a pumping schedule, folding candidates in one after another and keeping a merge only when it improves the end-to-end fidelity, so FOLD inherits the best of the three orders, and the order in which it visits the candidates matters little.

For a fixed number of paths $k_{max}$, FOLD runs in $\bigO\!\big(r\,(n^{2}+m)+n\,m\log n\big)$ time, where $n$ counts the nodes, $m$ the edges, and $r=m-n+c$ is the circuit rank with $c$ the number of connected components (\cref{thm:heur}). The circuit rank is the minimum number of edges whose removal breaks every cycle, so it measures both the density of the graph and the parallel edges within it. It also governs the running time: every successful merge consumes one independent cycle, so at most $r$ merges succeed (\cref{lem:rank}), and the work scales with the cycle structure of the graph. Dense graphs, where $r=\Theta(m)$, give $\bigO(mn^{2})$ (\cref{cor:dense}), and near-tree networks, where $r=\bigO(1)$, fall back to the $k$-shortest-paths computation alone. Either way this is polynomial, against the $m^{\Theta(m)}$ of the exact solver. Increasing the path budget $k_{\max}$ adds little to this cost, since it feeds the $k$-shortest-paths search but not the rank-bounded merging phase (\cref{fig:kmax_runtime}).

For simplicity the algorithm swaps linearly, one operation after another along the final path. To account for decoherence or latency, apply the swaps pairwise in a tree so that they run in parallel, mirroring the schedulable ``tree'' optimization problem \citet{jia2024routing} identify for distillation. The full routine, a step-by-step walkthrough and a runnable implementation are available in the paper's repository~\citep{entanglenet_code}. The companion visualizer replays these merges operation by operation on a network the reader builds~\citep{entanglenet_visualizer}.

Prior algorithms such as Q-PATH~\citep{li2022fidelity}, Q-CAST~\citep{shi2020concurrent}, DFER~\citep{tan2026distributed} and greedy routing~\citep{chakraborty2019distributed} optimize the entanglement-distribution rate in a continuously replenished network, not the fidelity attainable from one fixed set of resources, so we do not benchmark them against an objective they were never designed for. What separates FOLD from them, and from the multipath purification literature that also distills over node-disjoint paths~\citep{mondal2024multipath,bala2025statistical}, is where distillation happens: link-level schemes purify parallel copies of one edge and whole-path schemes combine two end-to-end paths, whereas ours purifies node-disjoint detours at multiple merge points and discovers those detours itself, as part of routing.

\subsection{Benchmark}
\label{sec:benchmark}
We benchmark FOLD against four simpler references on the $261$ Internet topologies of the Topology Zoo~\citep{knight2011zoo}, so the comparison runs on graphs operators have actually built. The benchmark simulates a quantum network laid over an actual Internet topology. We draw the entanglement on each graph probabilistically, every edge fidelity taken uniformly from a narrow band. These maps are far too large for the exhaustive solver, which reaches only small graphs, so where a question needs the exact optimum as a denominator we recover it from a separate set of smaller random samples of complex topologies. The benchmark scripts, the topology corpus and the per-instance results are in the paper's repository~\citep{entanglenet_code}.

Our experiments route on a simpler weight, $\hat{w}_e=-\log_2(2f_e-1)$, which the implementation reads directly from the negativity $N_e=\tfrac{2f_e-1}{2}$ of a Werner edge. It falls with $f_e$ exactly as $w_e$ does, so the two weights rank single edges identically, but only $w_e$ turns \cref{eq:path_fidelity} into an exact sum. The two weights select a different path in only about $2\%$ of the instances, and $\hat{w}_e$ costs less than $10^{-4}$ in mean end-to-end fidelity. FOLD accepts either weight, and $w_e$ is the one to use when the exact guarantee matters.

The references span the natural strategies. Shortest-path routing (SP) purifies nothing. Link-purify routing (LPR) distills every bundle of parallel edges into one link, then swaps along the shortest path of the resulting graph, so it never uses a detour. Two-best-path purification (TwoPath) swaps two disjoint paths down and distills the two results. The recursive reduction distills parallel links and collapses every degree-two relay until neither move applies, then routes whatever remains by shortest path~\citep{meng2023series, meng2021concurrence}. The companion visualizer runs all five strategies side by side, and the exact solver alongside them where the graph is small enough~\citep{entanglenet_visualizer}.

Purification is not always possible. Only a quarter of the source--target pairs we sample on the Topology Zoo admit any strategy that beats plain shortest-path routing, so \cref{fig:heuristic_benchmark}(a),(b) reports that purifiable subset alone. There FOLD leads every reference, the recursive reduction included, and it collects about twice the fidelity gain of the reduction, the strongest of the four. The exhaustive solver does not appear in these panels. A Zoo map carries a median of $34$ links, well past the size at which an exact answer is still reachable (\cref{fig:heuristic_benchmark}(e),(f)). That is the practical case for a heuristic in the first place.

The recursive reduction is the strongest reference. Its simplicity and its recursive use of the available resources make it a natural choice for a one-shot quantum network, and it reaches the optimum on every graph it is able to consume entirely. These are the \emph{series-parallel}, or \emph{reducible}, graphs~\citep{meng2023series}, the ones that merging parallel links and collapsing degree-two relays turn into a single source--target edge. Most of these topologies are not reducible, though, and there the reduction fails, which is where FOLD outperforms it. On our random samples, small enough to solve exhaustively, the reduction recovers most of the optimal fidelity gain on reducible graphs but almost none on non-reducible ones, where FOLD still recovers most of it (\cref{fig:heuristic_benchmark}(d)). Deployed topologies fall mostly on the non-reducible side, since a topology stops being reducible past a few independent cycles, so the reduction completes on only a fifth of the purifiable Zoo pairs and almost none past circuit rank $10$ (\cref{fig:heuristic_benchmark}(c)).

We do not read this as FOLD dominating everywhere. Where a network is reducible, the simpler reduction remains the better choice, and our own benchmark shows it. We read it instead as evidence that deployed topologies mostly fall outside that case, so the regime in which FOLD outperforms the reduction is the one a real network is likely to present.

FOLD also runs in polynomial time, against the super-exponential cost of the exact solver. \cref{fig:heuristic_benchmark}(e),(f) shows the median running time of all six strategies on Erd\H{o}s--R\'enyi graphs, against graph size and against circuit rank. The five polynomial strategies never differ by more than a factor of $3.2$, so FOLD costs no more than the simplest reference it beats. The exact solver is among the cheapest of the six on a tree, and its cost then grows beyond the range of the panel. Each further independent cycle multiplies its cost by about an order of magnitude, from $6$ ms at circuit rank $0$ to $320$ s at rank $5$, past which it no longer finishes within an hour. The rank that puts an exact answer out of reach is the same rank that leaves the reduction with nothing to offer, and a deployed backbone sits far above both.

\begin{figure}[H]
  \centering
  \includegraphics[width=0.95\columnwidth]{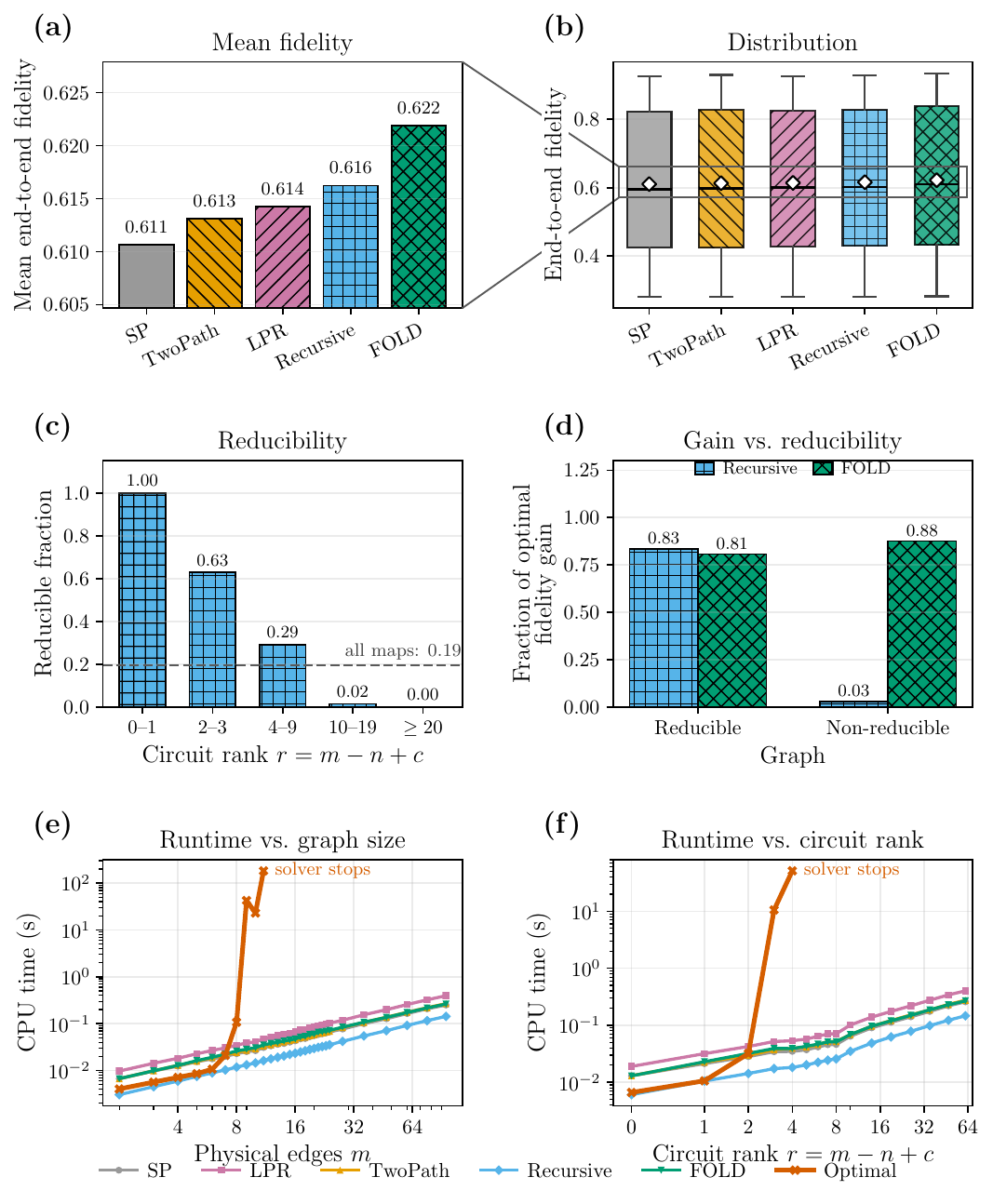}
  \caption{\textbf{FOLD leads every reference on Internet-derived
  topologies while running as fast as the cheapest of them.}
  The purifiable source--target pairs of the Internet topologies of the Topology
  Zoo~\citep{knight2011zoo} feed \textbf{(a)}--\textbf{(c)}, random samples feed
  \textbf{(d)}, and Erd\H{o}s--R\'enyi graphs feed \textbf{(e)},\textbf{(f)}.
  \textbf{(a)},\textbf{(b)} Mean end-to-end fidelity and its distribution
  (box${}={}$IQR, line${}={}$median, $\diamond={}$mean, whiskers${}={}$5--95\%)
  for shortest-path routing (\emph{SP}), link-purify routing
  (\emph{LPR}), two-best-path purification (\emph{TwoPath}), the recursive
  reduction (\emph{Recursive}) and \emph{FOLD}, keyed by color and fill
  pattern and ordered by increasing fidelity. The exhaustive \emph{Optimal} is
  out of reach at this size.
  \textbf{(c)} How often the reduction completes, against circuit rank. The
  dashed line marks the rate over the whole corpus.
  \textbf{(d)} The fraction of the optimal fidelity gain,
  $(F-F_{\mathrm{SP}})/(F_{\mathrm{opt}}-F_{\mathrm{SP}})$, that the reduction and
  FOLD recover on reducible and non-reducible graphs.
  \textbf{(e)},\textbf{(f)} Median CPU time against graph size \textbf{(e)} and
  against circuit rank \textbf{(f)}. Both axes are logarithmic, the rank axis
  symmetrically so, since a tree has rank zero. Each exhaustive curve stops
  where the solver no longer finishes within an hour.}
  \label{fig:heuristic_benchmark}
\end{figure}

\section{Noise and Uncertainty}
\label{sec:uncertainty}

So far we have treated the edge fidelities $\{f_e\}$ as known exactly. In practice they are not. A single copy of a state does not reveal its fidelity, so the network estimates it from many identical copies, and that estimate degrades while the state waits in memory \citep{delledonne2025operating}. A router thus commits to a structure on the mean fidelities, and the realized fidelities determine what it delivers. A solution must remain optimal once the fidelities are realized, not merely report a high fidelity. This matters most at the boundary of the purifiable region of \cref{fig:distillation_selectivity}(a), where a small error in $f_e$ flips the decision of whether a distillation helps at all.

We therefore give each edge a distribution over its fidelity. Tomography is a Bernoulli trial between $\ket{\Psi^-}$ and the totally mixed state (\cref{eq:first_werner_form}), so a Beta distribution is the natural belief about one edge. Swapping multiplies depolarizing parameters, which makes $X_e=-\log p_e$ the coordinate that adds along a path and the Gamma the family closed under that addition. We place a Gamma on each $X_e$ and pin its two parameters so that the mean of an edge is the fidelity the router was handed, $\mathbb{E}[f_e]=\bar f_e$, and every edge is characterized to the same precision, $\mathrm{sd}(f_e)=\sigma_F$. A single parameter $\sigma_F$ sets the uncertainty of the whole network. Raising it adds spread, leaves every mean unchanged and never produces an unphysical state. Edges of the same link stay independent, as each is physically distinct.

\begin{figure}[H]
  \centering
  \includegraphics[width=\columnwidth]{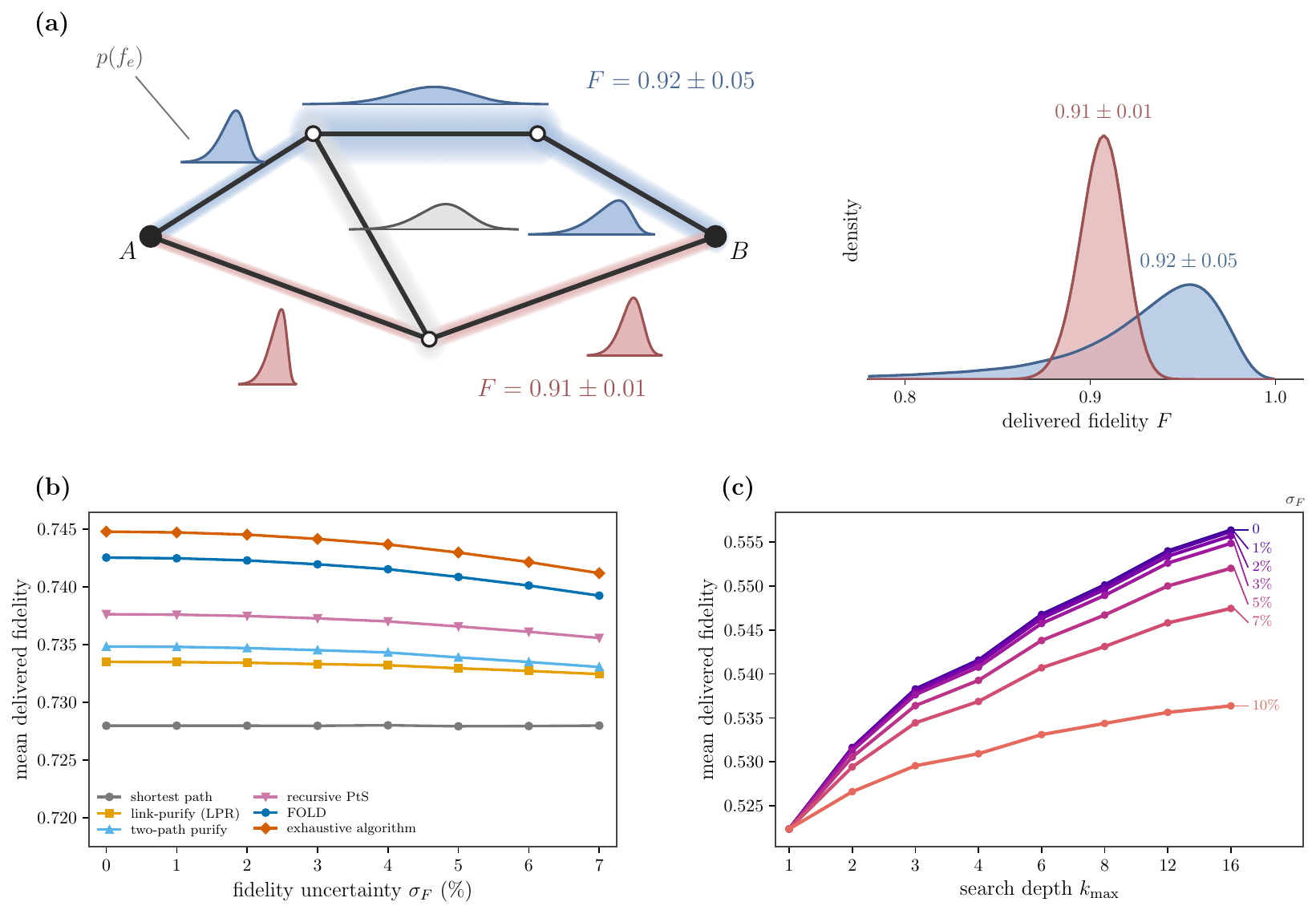}
  \caption{\textbf{Routing on the mean fidelities: the cost of uncertainty and the
  effect of search depth.}
  \textbf{(a)} Each link carries a distribution over its fidelity rather than a
  number, and a route's highlight blurs with the spread of its links. The upper
  route has the higher mean fidelity, $F=0.92\pm0.05$, and the disjoint lower
  route less uncertainty in its end-to-end fidelity, $F=0.91\pm0.01$ (propagation rule in
  \cref{sec:uncertainty_propagation}).
  \textbf{(b)} Mean delivered fidelity against the per-edge uncertainty
  $\sigma_F$ for the six strategies of \cref{fig:heuristic_benchmark} on the
  purifiable graphs, each committing on the mean graph and executing on
  thousands of realized ones. The ordering of the strategies holds at every
  $\sigma_F$, FOLD stays within $0.0023$ of the exhaustive solver, and only
  shortest-path routing is flat.
  \textbf{(c)} Mean delivered fidelity against the search depth $k_{\max}$ on
  $150$-node networks with source and target about $15$ hops apart, one curve
  per uncertainty level. The $\sigma_F=0$ curve is the case of no uncertainty,
  where the fidelity the algorithm computes is the one it delivers. Every other
  curve falls below it. Deeper search raises the
  mean-graph fidelity by $0.034$ between $k_{\max}=1$ and $16$, of which $98\%$ is
  delivered at $\sigma_F=2\%$ and $41\%$ at $10\%$}
  \label{fig:uncertainty_benchmark}
\end{figure}

\subsection{Routing Under Uncertainty}
\label{sec:routing_under_uncertainty}

We analyze how FOLD performs when the fidelities are not known exactly, and how it compares to existing approaches. Each algorithm decides on the mean graph. We then execute that decision on thousands of graphs drawn from the distribution model and record the delivered fidelity. 
\Cref{fig:uncertainty_benchmark}(a) shows the difficulty on a small example: the router commits a route or scheduling, the better one on the mean graph, while a different strategy might be the better one once the fidelities are realized. 
\Cref{fig:uncertainty_benchmark}(b) applies this experiment to all six strategies of \cref{sec:benchmark}. Every strategy that purifies loses fidelity as $\sigma_F$ grows, and FOLD loses at the same rate as the exhaustive solver, keeping its margin over the recursive reduction and staying within $0.0023$ of the optimum on the mean graph at every noise level. Shortest-path routing is the exception and stays flat, since it purifies nothing and its structure is a pure swap chain, along which expectations multiply exactly. The ordering of the strategies is a property of the topology and is unchanged by the spread, and the uncertainty affects only the strategies that distill.

Uncertainty costs more than any parameter the routing layer controls. Solving the routing exactly rather than with FOLD is worth $0.0015$ in delivered fidelity, whereas not knowing the fidelities to better than $10\%$ costs $0.039$ against a router that already knows the realized values, some $26$ times more. At that noise level the structure chosen on the mean graph is still the best one on only two of every five realized graphs, and doubling FOLD's search depth recovers none of the loss. \Cref{fig:uncertainty_benchmark}(c) traces the whole trade-off. Raising $k_{max}$ from $1$ to $16$ lifts the fidelity promised on the mean graph by $0.034$, but the delivered fidelity falls below that value as $\sigma_F$ grows, $98\%$ of that gain arriving at $\sigma_F=2\%$, $87\%$ at $5\%$ and only $41\%$ at $10\%$. Under uncertainty the binding constraint is information, not search depth.
\section{Multiple Requests}
\label{sec:requests}
In a one-shot quantum network, several pairs that wish to communicate
simultaneously must share the same resources, which creates a trade-off between
the fidelities they can each achieve. Using the optimal algorithm described in
\cref{sec:optimal_brute_force_algorithm}, we analyze this multiple-request
scenario by plotting a Pareto graph of the achieved fidelities across all
sequences of operations (\cref{fig:multiple_requests}). We also show how FOLD
(\cref{sec:heuristic_algorithm}) can serve multiple requests by interleaving them.

Some algorithms, such as \citet{li2021routing}, also optimize for multiple
requests, but treat fidelity not as the objective but as a constraint. They
maximize a weighted network flow and report link utilization, delay and Jain's
fairness index across the served requests. \citet{koutsopoulos2024}
maximizes a utility of the end-to-end fidelities of the form
$u(\mathbf{x},\mathbf{F})=\log\sum_i x_i g(F_i)$, following the network utility
maximization framework of \citet{vardoyan2022utility}, in which the objective is
averaged over many delivered pairs. The Pareto frontier we plot instead keeps the
requests separate, so that the trade-off between them is visible rather than
aggregated away.

One network carries this analysis, small enough that enumerating every schedule
remains possible. What follows is a proof of concept for routing several
requests at once, and the effects it shows call for a broader study across
topologies.

\Cref{fig:multiple_requests} shows a graph with $6$ nodes and $9$ links, together
with its Pareto graph, which plots every possible sequence of operations as a
function of the achieved fidelity for request A ($A_1$--$A_2$) and request B
($B_1$--$B_2$) within that graph. Requests A and B share the network
(\cref{fig:multiple_requests}(a)), and the three parallel $A_1$--$B_1$ links are
the contested resource. Request A purifies over as many of them as its budget
allows and swaps at $B_1$, and request B does the same and swaps at $A_1$. Each
falls back on a detour through its own relay once it loses them. Exchanging
$A_1$ with $B_1$, $A_2$ with $B_2$ and $R_1$ with $R_2$ maps every link onto a
link of equal fidelity, so neither request owns what the two compete for. The
Pareto frontier in \cref{fig:multiple_requests}(b) bounds the region the two
requests can jointly reach, and that exchange makes the region symmetric about
$F_A=F_B$.

The optimal frontier assumes a coordinator that plans both requests jointly,
whereas FOLD (\cref{sec:heuristic_algorithm}) serves a single pair. The
natural way to serve two requests is to run FOLD independently for each
one. However, this independent approach fails as each run greedily consumes physical links,
so the two solutions reuse the same links and cannot be delivered together. Here
both runs claim all three contested links and report $0.839$, an outcome that
lies outside the jointly achievable region.
Coordination between requests is therefore a condition for feasibility, not merely
a route to higher fidelity.

We coordinate the two requests through a shared path budget. Both requests draw
from a single budget of $k_{\max}$ candidate paths, and each request takes one
path at a time in alternation, routing on the graph left by the other. The
outcome is always feasible by construction. In this example, interleaving reaches
the optimal Pareto frontier at a small fraction of the exhaustive solver's cost
(\cref{fig:multiple_requests}(b)). The
division of the budget between the requests shifts the operating point along the
frontier, trading one request's fidelity against the other's. Returns diminish
quickly. Once a request saturates after a few paths, further paths add nothing,
so a modest shared budget serves both requests.

How the two requests take turns spending the shared budget also matters, not
just how the budget is divided between them. Interleaving turn by turn balances
the two fidelities and reaches the frontier at $(0.834, 0.820)$,
whereas block allocation, in which one request spends its entire share before
the other begins, drives the outcome to the frontier's extreme
(\cref{fig:multiple_requests}(b)). The request served first then keeps its
fidelity while the second is left with the least the frontier allows. At an even
split, block allocation costs the second request $0.044$ and the first one
$0.007$. Contention for physical links, rather than the number of paths, binds
the outcome, and the budget split acts in this example as a fairness control
over the frontier.

\begin{figure}[H]
  \centering
  \includegraphics[width=\columnwidth]{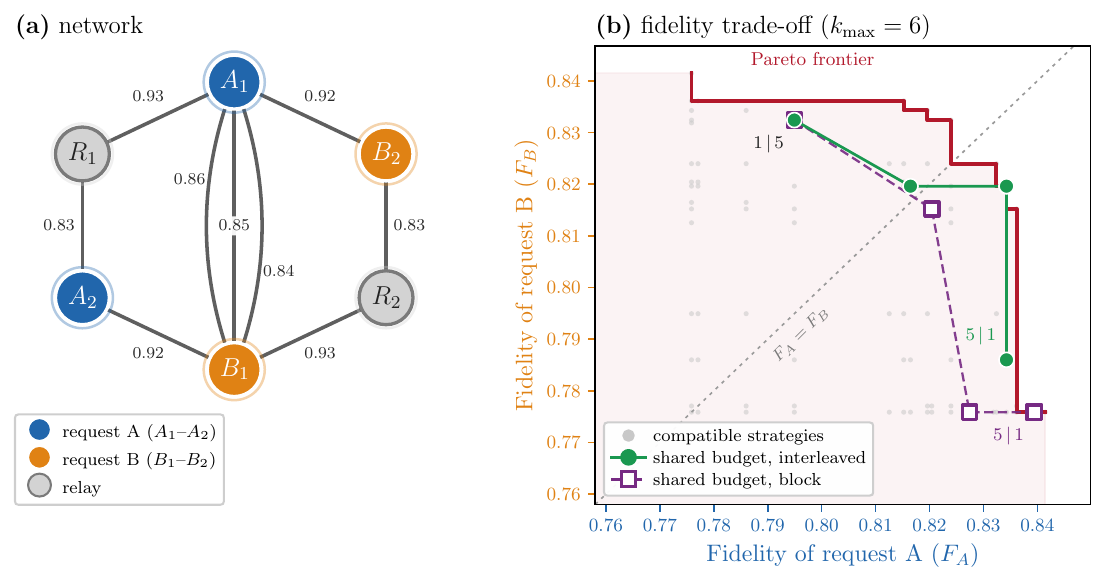}
  \caption{\textbf{Two simultaneous requests, from the optimum to
  FOLD.} Both panels use one network of $6$ nodes and $9$ links.
  \textbf{(a)} The network. Request A ($A_1$--$A_2$, blue) and request B
  ($B_1$--$B_2$, orange) share it, $R_1$ and $R_2$ relay, and the three
  parallel $A_1$--$B_1$ links are the contested resource. Exchanging the two
  requests maps the network onto itself, so neither of them owns those links.
  \textbf{(b)} The
  fidelity trade-off. Grey points are every compatible pair of operation
  sequences, and the red staircase is the optimal Pareto frontier of the
  discrete set of joint outcomes, symmetric about the dashed $F_A=F_B$ line.
  Under a shared budget $k_{\max}=6$, interleaving
  the two requests turn by turn (green) reaches the frontier at a balanced point,
  and the budget split shifts that point along the frontier. Serving one request
  fully before the other (block, purple) drives the outcome to the frontier's
  extreme, maximizing the first request at the expense of the second. Twenty of
  the $97$ compatible pairs cripple one request and fall below the frame.}
  \label{fig:multiple_requests}
\end{figure}
\section{Conclusions}
\label{sec:conclusions}

We recast entanglement routing for the one-shot regime that distributed
quantum computation demands. There a fixed set of pre-distributed pairs must be
combined into the highest end-to-end fidelity before coherence is lost, rather
than a rate of threshold-passing pairs accumulated over time. Optimizing this
objective across the whole topology reveals that the purify-then-swap ordering
much of the literature assumes is not optimal network-wide, and that the exact
solution does not scale, costing $m^{\Theta(m)}$ time. FOLD answers both problems
at once (\cref{sec:routing}). It selects a multi-path and the schedule that
combines it in polynomial time, recovers almost all of the optimal fidelity, and
runs four orders of magnitude faster than the exact solver at circuit rank $5$.
On the topologies of real Internet networks, FOLD leads every reference we test,
and its advantage is largest on the non-reducible graphs that deployed networks
present.
We also apply the same model to uncertain edge fidelities, where information
rather than search depth becomes the binding constraint
(\cref{sec:uncertainty}). We illustrate its extension to several simultaneous
requests as well (\cref{sec:requests}). There a shared path budget keeps the
joint solution feasible, and in our example it controls the trade-off along the
Pareto frontier. Our reported fidelities are conditional on successful
completion of the selected schedule, and they do not account for the
probability of reaching that outcome. Extending the routing objective to
balance conditional fidelity against delivery probability, and designing
adaptive recovery after failed operations, remain important directions for
future work. Together these
results argue for treating routing as a fidelity-optimization problem over the
entire network, and they point toward a routing layer that a quantum internet
can run at scale, wherever many processors must hold entanglement together to
compute as one.

\bibliographystyle{unsrtnat}
\bibliography{refs}

\BeginPart{supp}
\appendix
\crefalias{section}{appendix}
\numberwithin{equation}{section}
\phantomsection
\addcontentsline{toc}{section}{Supplemental Material}

\begin{center}
  {\LARGE\textbf{Supplemental Material}}\\[0.4em]
  {\large for ``One-shot Routing in Quantum Networks''}
\end{center}
\vspace{1.5em}

\section{Entanglement Operations: Swapping and Distillation}
\label{sec:operations}
Two operations recursively combine entangled states in a one-shot
network: entanglement swapping \citep{zukowski1993swapping} and entanglement
distillation \citep{bennett1996purification}
(\cref{fig:operations}(a)). Swapping can establish entanglement between nodes
with no direct edge, whereas distillation consumes parallel states between the
same endpoints to produce a higher-fidelity state.
Both consume existing Werner states and produce a new state whose fidelity follows from the fidelities of the states they consume.
The output of either operation is itself an entangled state, so it can serve as an input to a later swap or distillation. We use lowercase $f_i$ for the fidelity of an elementary edge and uppercase $F_i$ for the fidelity of a state produced by earlier operations. Thus, $\Fpur(F_1,F_2)$ denotes the distillation of two previously constructed states, as in \cref{alg:brute_force}.

\subsection{Swapping}
Swapping joins two states through an intermediate node and destroys both, and the new state inherits the noise of each.
Sequencing it along a path of entangled states entangles the source and target of that path.
Two states of fidelities $f_1,f_2$ swap into
\begin{equation}
    \Fswap(f_1,f_2) = f_1f_2 + \frac{(1-f_1)(1-f_2)}{3},
\end{equation}
which \cref{eq:swapping_fidelity} of the main text also states.
Writing each fidelity through its depolarizing parameter $p=\tfrac{4f-1}{3}$ turns the relation into a plain product,
\begin{equation}
    p_{\mathrm{swap}} = p_1p_2,
\end{equation}
since substituting $f=\tfrac{1+3p}{4}$ on both sides gives $\Fswap=\tfrac{1+3p_1p_2}{4}$.
A path $P$ therefore accumulates $\prod_{e\in P}p_e$ across its edges, and inverting the substitution recovers \cref{eq:path_fidelity} of the main text.

\subsection{Distillation}
Distillation is likewise defined for two states, and applying it in sequence raises one state out of several.
Two states of fidelities $f_1,f_2$ purify into
\begin{equation}
    \Fpur(f_1,f_2) = \frac{f_1f_2+\frac{1}{9}(1-f_1)(1-f_2)}{f_1f_2+\frac{1}{3}(f_1+f_2-2f_1f_2)+\frac{5}{9}(1-f_1)(1-f_2)},
\end{equation}
the BBPSSW relation \citep{bennett1996purification} that \cref{eq:distillation_fidelity} of the main text states.
Unlike swapping, it succeeds only probabilistically.
No substitution linearizes it the way the depolarizing parameter linearizes swapping, which is what costs \cref{sec:routing} its optimal substructure.

\section{Uncertainty Propagation Along a Path}
\label{sec:uncertainty_propagation}

Swapping collapses a path into a single edge through \cref{eq:path_fidelity}, so the uncertainty of the edges reaches the end-to-end fidelity through that same product. Let $p_e=\tfrac{4f_e-1}{3}$ be the depolarizing parameter of edge $e$, now a random variable of mean $\mathbb{E}[p_e]$ and standard deviation $\mathrm{sd}(p_e)$. A path uses each link at most once, so its edges are independent and both moments factor over it,
\begin{align}
  \mathbb{E}\big[F(P)\big] &= \frac{1}{4}\Big(1+3\prod_{e\in P}\mathbb{E}[p_e]\Big),
  \label{eq:path_mean}\\
  \mathrm{Var}\big(F(P)\big) &= \Big(\mathbb{E}\big[F(P)\big]-\frac{1}{4}\Big)^{2}
    \Big(\prod_{e\in P}\big(1+c_e^{2}\big)-1\Big),
  \label{eq:path_variance}
\end{align}
where $c_e=\mathrm{sd}(p_e)/\mathbb{E}[p_e]=\mathrm{sd}(f_e)/\big(\mathbb{E}[f_e]-\tfrac14\big)$ is the coefficient of variation, the uncertainty of an edge relative to its distance above the fully mixed floor $F=\tfrac14$. \Cref{eq:path_mean} is \cref{eq:path_fidelity} evaluated at the mean fidelities, so the mean of each edge propagates exactly and needs no correction term.

The variance carries the whole effect. Dividing \cref{eq:path_variance} by that same square defines the relative uncertainty $c(P)=\mathrm{sd}\big(F(P)\big)\big/\big(\mathbb{E}[F(P)]-\tfrac14\big)$ of the delivered state,
\begin{equation}
  c(P)^{2}=\prod_{e\in P}\big(1+c_e^{2}\big)-1\;\approx\;\sum_{e\in P}c_e^{2},
  \label{eq:relative_uncertainty}
\end{equation}
the approximation holding whenever every $c_e\ll1$. That limit is what the first-order propagation rule $\sigma_y^{2}=\sum_e(\partial y/\partial x_e)^{2}\sigma_{x_e}^{2}$ returns for \cref{eq:path_fidelity}, whose derivative $\partial F(P)/\partial f_e=\prod_{e'\in P\setminus\{e\}}p_{e'}$ is the path product with edge $e$ removed. Relative uncertainties add in quadrature along a path, and \cref{eq:relative_uncertainty} states exactly how much the true product exceeds that sum.

The Gamma of \cref{sec:routing_under_uncertainty} makes the growth concrete. Take a path of identical edges, $X_e=-\log p_e\sim\Gamma(k,\theta)$. Its moment-generating function evaluated at $-1$ and $-2$ gives $\mathbb{E}[p_e]=(1+\theta)^{-k}$ and $\mathbb{E}[p_e^{2}]=(1+2\theta)^{-k}$, so a path of $\ell$ edges carries
\begin{equation}
  c(P)^{2}=\left(\frac{(1+\theta)^{2}}{1+2\theta}\right)^{k\ell}-1 .
  \label{eq:gamma_relative_uncertainty}
\end{equation}
The base exceeds $1$ for every $\theta>0$, so relative uncertainty compounds geometrically with the hop count, at the same rate at which the mean fidelity decays. A long path is doubly penalized, even though its absolute spread eventually shrinks together with the fidelity itself.

Propagation also settles the form of the risk term. In the log coordinate the additive risk measure is the variance, $\mathrm{Var}(X_{\text{path}})=\sum_{e\in P}\mathrm{Var}(X_e)$, whereas the edge weight of \cref{sec:routing_under_uncertainty} sums standard deviations,
\begin{equation}
  \sum_{e\in P}\mathrm{sd}(X_e)\;\ge\;\sqrt{\textstyle\sum_{e\in P}\mathrm{Var}(X_e)}\;=\;\mathrm{sd}(X_{\text{path}}),
  \label{eq:risk_superadditive}
\end{equation}
with equality only when a single edge carries all the variance, and with a ratio of $\sqrt{\ell}$ when the edges share one distribution. The additive weight therefore charges a path up to $\sqrt{\ell}$ times the risk it actually carries, a tax on hops beyond the mean weight. That tax is what secures the tail gain of the risk term, and explains why a large $\lambda$ overshoots. Penalizing $\mathrm{Var}(X_e)$ instead reproduces a mean-variance objective exactly, but its units differ from the mean's, so we keep the standard deviation and let $\lambda$ absorb the mismatch.

\section{Optimal Brute-Force Algorithm: Construction and Complexity}
\label{supp-mat:brute-force}
This section describes in more detail how the algorithm works, followed by pseudo-code of the implementation. The full implementation is available in the EntangleNet repository on GitHub~\citep{entanglenet_code}.
The section continues with a complexity proof of the algorithm.
\subsection{Pseudo-Code and Walkthrough}

The algorithm manipulates a single growing list $L$ of \emph{virtual edges}. Each virtual edge $v$ is a record of everything needed to reason about a candidate sequence of operations: its endpoints $\mathrm{ends}(v)$, the set $\mathrm{used}(v)\subseteq E$ of \emph{original} graph edges it consumes, an action encoding $\mathrm{enc}(v)$, and the resulting fidelity $F(v)$. The list is seeded with the original edges of the graph, each occupying a single edge index and encoded by that index. New virtual edges are then built by combining two earlier ones, and the encoding is composed recursively so that every entry records the full operation tree that produced it: a distillation of $v_1$ and $v_2$ is written $(\textproc{purify},\ \mathrm{enc}(v_1),\ \mathrm{enc}(v_2))$ and a swap $(\textproc{swap},\ \mathrm{enc}(v_1),\ \mathrm{enc}(v_2))$, which unfolds into nested expressions such as $(\textproc{swap},\ (\textproc{purify},\ i,\ j),\ k)$ for a two-hop path whose first hop was distilled from edges $i$ and $j$. In the implementation this tree is realized as a compact recursive tag keyed by the consumed edge indices. Equivalently, every virtual edge is a rooted binary tree whose leaves are distinct original edges and whose internal nodes are purify/swap primitives, each constrained by the compatibility rule of its two children.

The construction proceeds through two nested loops: the outer loop walks the list, and for each virtual edge $v_i$ the inner loop revisits every earlier virtual edge $v_j$ and tries to combine the pair. A pair is combinable only if $\mathrm{used}(v_i)\cap\mathrm{used}(v_j)=\emptyset$, since a physical entangled link may serve at most one operation. Tracking $\mathrm{used}(\cdot)$ as an index set (a bitmask in the implementation) makes this check a single disjointness test. Given a disjoint pair, the topology of their endpoints fixes the operation: if the two virtual edges connect the same pair of nodes they are parallel and are merged by distillation, whereas if they share exactly one node they are concatenated by entanglement swapping. The combined fidelity follows the corresponding relation, and its $\mathrm{used}$ set is the union of the two parents'. Some works instead optimize over the available protocols for swapping and distillation \citep{goodenough2021,rozpkedek2018optimizing}. Because each newly formed virtual edge is appended to the end of $L$, it will itself be revisited as a candidate for further combination, so the loop eventually enumerates every valid sequence of operations reachable on the graph. Throughout, a table $D$ keeps, for every pair of nodes, the highest-fidelity virtual edge seen so far, so the final answer is read off directly once the list is exhausted. Virtual edges that already consume all of $E$ cannot be extended and are set aside in $S$ rather than re-appended to $L$.

\begin{algorithm}[H]
\caption{Optimal Brute-Force Encoding Algorithm}
\label{alg:brute_force}
\begin{algorithmic}[1]
\Require Quantum network $G=(V,E)$ with edge fidelities $\{f_e\}_{e\in E}$
\Ensure Table $D$ of the optimal virtual edge for every pair of nodes
\Statex
\Function{Encode}{$G$}
  \State $L \gets \langle\,\rangle$ \Comment{ordered list of virtual edges}
  \State $S \gets \emptyset$ \Comment{virtual edges consuming all of $E$}
  \State $D[u][w] \gets \textsc{Null}$ for all $u,w \in V$
  \ForAll{$e=(u,w) \in E$} \Comment{initialize with the original edges}
    \State $v \gets$ virtual edge with $\mathrm{ends}(v)=\{u,w\}$, $\mathrm{used}(v)=\{e\}$, $F(v)=f_e$
    \State append $v$ to $L$
    \State \Call{Update}{$D,v$}
  \EndFor
  \Statex
  \For{$i \gets 1$ \textbf{to} $|L|$} \Comment{$L$ grows while iterated}
    \State $v_i \gets L[i]$
    \For{$j \gets 1$ \textbf{to} $i-1$}
      \State $v_j \gets L[j]$
      \If{$\mathrm{used}(v_i) \cap \mathrm{used}(v_j) \neq \emptyset$}
        \State \textbf{continue} \Comment{each edge may serve a single operation}
      \EndIf
      \If{$\mathrm{ends}(v_i) = \mathrm{ends}(v_j)$} \Comment{parallel: distillation}
        \State $v \gets$ \Call{Combine}{$v_i,v_j,\textproc{purify}$}
      \ElsIf{$|\mathrm{ends}(v_i) \cap \mathrm{ends}(v_j)| = 1$} \Comment{shared node: swapping}
        \State $v \gets$ \Call{Combine}{$v_i,v_j,\textproc{swap}$}
      \Else
        \State \textbf{continue} \Comment{not combinable}
      \EndIf
      \If{$\mathrm{used}(v) \neq E$}
        \State append $v$ to $L$ \Comment{may still be extended further}
      \Else
        \State $S \gets S \cup \{v\}$ \Comment{no edges left to combine with}
      \EndIf
      \State \Call{Update}{$D,v$}
    \EndFor
  \EndFor
  \State \Return $D$
\EndFunction
\algstore{bruteforce}
\end{algorithmic}
\end{algorithm}

\begin{algorithm}[H]
\addtocounter{algorithm}{-1}
\caption{Optimal Brute-Force Encoding Algorithm (continued)}
\begin{algorithmic}[1]
\algrestore{bruteforce}
\Function{Combine}{$v_1,v_2,\mathit{action}$}
  \State $\mathrm{used}(v) \gets \mathrm{used}(v_1) \cup \mathrm{used}(v_2)$
  \State $\mathrm{enc}(v) \gets (\mathit{action},\ \mathrm{enc}(v_1),\ \mathrm{enc}(v_2))$
  \If{$\mathit{action} = \textproc{purify}$}
    \State $\mathrm{ends}(v) \gets \mathrm{ends}(v_1)$
    \State $F(v) \gets \Fpur\big(F(v_1),F(v_2)\big)$ \Comment{\cref{eq:distillation_fidelity}}
  \Else
    \State $m \gets$ the single node in $\mathrm{ends}(v_1) \cap \mathrm{ends}(v_2)$
    \State $\mathrm{ends}(v) \gets \big(\mathrm{ends}(v_1) \cup \mathrm{ends}(v_2)\big) \setminus \{m\}$
    \State $F(v) \gets \Fswap\big(F(v_1),F(v_2)\big)$ \Comment{\cref{eq:swapping_fidelity}}
  \EndIf
  \State \Return $v$
\EndFunction
\Statex
\Function{Update}{$D,v$}
  \State $\{u,w\} \gets \mathrm{ends}(v)$
  \If{$D[u][w] = \textsc{Null}$ \textbf{or} $F(D[u][w]) < F(v)$}
    \State $D[u][w] \gets v$; \quad $D[w][u] \gets v$
  \EndIf
  \EndFunction
\end{algorithmic}
\end{algorithm}

\subsection{Complexity}
In this part we would prove a worst-case super-exponential complexity of the algorithm. However, the per-instance complexity ranges from polynomial (star) through exponential (long path) to super-exponential ($G_m$).
We first start by proving a general numbering of the trees \cref{lem:doublefact}. In \cref{thm:brute} we then show this can be both the general upper bound of the algorithm, as well as the lower bound for a certain type of graphs $G_m$. Therefore giving a tight worst-case running time complexity of the exhaustive algorithm. The two remaining points of that range are settled in \cref{prop:star_path}.

\paragraph{Search Space}

The brute-force algorithm performs a \emph{closure
computation}. It initializes a working list with the $m$ physical edges and then
repeatedly appends the combination of every \emph{earlier} pair of list elements
that is (i)~edge-disjoint and (ii)~\emph{compatible}, i.e.\ shares both endpoints
(eligible for \texttt{purify}) or exactly one endpoint (eligible for
\texttt{swap}). Every object it constructs is therefore a rooted binary tree whose
leaves are distinct physical edges and whose internal nodes are primitives, subject
to the compatibility constraint at each internal node. The two structures should not be conflated: a \emph{leaf} of an operation tree is a
physical \emph{edge} of the network $G$, while an internal node of the tree is an
operation rather than a network node. A tree with $s$ leaves thus consumes $s$
physical links of $G$.

\begin{definition}[Operation trees]
\label{def:optree}
Let $\Tree(G)$ be the set of all full binary trees whose leaves are pairwise
distinct physical edges of $G$ and in which every internal node combines two
edge-disjoint, compatible subtrees by a single primitive. Write $L=\lvert\Tree(G)\rvert$.
A tree in $\Tree(G)$ encodes one routing together with one order of operations.
\end{definition}

\begin{lemma}[Exhaustiveness]
\label{lem:exhaust}
The algorithm constructs every element of $\Tree(G)$ exactly once, and hence
materializes precisely $L$ virtual edges.
\end{lemma}

\begin{proof}
We argue by induction on the leaf-count $s$ of a tree $t\in\Tree(G)$. If $s=1$
then $t$ is a physical edge, inserted during initialization. If $s\ge 2$, the
root of $t$ partitions its leaves into two edge-disjoint subtrees $t_{L},t_{R}$,
each with fewer leaves and each lying in $\Tree(G)$, so by the induction hypothesis
both appear in the working list. Because the inner loop pairs each newly inserted
element with \emph{every} earlier element, and $(t_{L},t_{R})$ are edge-disjoint
and compatible, their combination, namely $t$, is generated. It is generated
exactly once: the loop visits only ordered index pairs $(i,j)$ with $j<i$, and the
primitive is symmetric in its operands, so each unordered pair is combined a single
time.
\end{proof}

\paragraph{Lower and Upper Bounds}

The lower bound rests on a classical enumeration, which we record for completeness.
The count below is \emph{unconstrained}: it enumerates all full binary trees on a
leaf set, imposing no compatibility requirement at the internal nodes. Its role in
the two directions of Theorem~\ref{thm:brute} is correspondingly different, and we
make this explicit in the proof.

\begin{lemma}[Number of operation trees on a fixed leaf set]
\label{lem:doublefact}
Fix a set of $s\ge 1$ labeled leaves. The number of full binary trees on these
leaves in which each internal node combines its two children by a single
\emph{commutative} binary operation, with no further restriction, is the double
factorial%
\footnote{%
  Here $n!!$ denotes the skip-product $n(n-2)(n-4)\cdots$, terminating at $2$ or
  $1$. It is a single operator, not an iterated factorial: $n!!\neq(n!)!$. For
  instance $5!!=5\cdot3\cdot1=15$.%
}
\begin{equation}
  T(s)\;=\;(2s-3)!!\;=\;1\cdot 3\cdot 5\cdots(2s-3),
  \qquad T(1)=T(2)=1 .
\end{equation}
\end{lemma}

\begin{proof}
Build the tree by inserting the leaves one at a time. Since the tree must remain
full, a new leaf cannot be hung directly on an existing node. It must instead be paired with
an existing subtree beneath a freshly created parent. Attachment sites are therefore
in bijection with subtrees, hence with \emph{nodes}: pairing with the subtree rooted
at a non-root node $v$ subdivides the edge above $v$, and pairing with the entire
tree caps a new root. A full binary tree on $j-1$ leaves has $j-2$
internal nodes, so $2j-3$ nodes in total and thus $2j-3$ sites. Every insertion adds exactly
two nodes -- the new leaf together with its new parent -- so the number of sites grows
by two at each step and, starting from $1$, runs through the odd integers
$1,3,5,\dots,2s-3$. The construction is a bijection:
deleting leaf $j$ from any $j$-leaf tree and smoothing its now degree-two parent
recovers a unique $(j-1)$-leaf tree together with the site the leaf occupied, so no
tree is counted twice or omitted. Multiplying the site counts over $j=2,\dots,s$
(leaf $2$ joins leaf $1$ in a single way) yields
\begin{equation}
  T(s)\;=\;\prod_{j=2}^{s}(2j-3)\;=\;1\cdot 3\cdot 5\cdots(2s-3)\;=\;(2s-3)!!,
\end{equation}
a product of consecutive odd integers, i.e.\ the double factorial. For instance,
$T(1),\dots,T(5)=1,1,3,15,105$.
\end{proof}

Note that \cref{lem:doublefact} counts all possible trees of such size and connections. It does not account for the constraints of physically applying swapping and distillation, as it counts all existing internal nodes as valid connections for the newly added leaf.
We now show in \cref{thm:brute} that such a count of all internal nodes as valid connections is possible for a certain type of graph $G_m$, and therefore the count is a lower bound on the number of operation trees in $\Tree(G_m)$, which then gives us a lower bound for the worst-case complexity of the algorithm. We then show that the same count is also an upper bound on the number of operation trees in $\Tree(G)$ for any graph $G$, and therefore also an upper bound on the worst-case complexity of the algorithm.

\begin{theorem}[Complexity of the exhaustive scheme]
\label{thm:brute}
The encoding algorithm materializes $L$ virtual edges and runs in $\Theta(L^{2})$
time. In the worst case
\begin{equation}
  L \;=\; 2^{\Theta(m\log m)} \;=\; m^{\Theta(m)},
\end{equation}
so the algorithm is super-exponential in the number of physical edges.
\end{theorem}

\begin{proof}
\emph{Lower bound.}
Consider the two-node multigraph $G_{m}$ consisting of nodes $u,v$ joined by $m$
parallel edges. We claim that on $G_{m}$ both restrictions imposed by the algorithm,
compatibility and edge-disjointness, are vacuous, so that the unconstrained count of
Lemma~\ref{lem:doublefact} is \emph{attained} rather than merely an over-estimate.

\emph{Compatibility.} Every virtual edge constructible on $G_{m}$ joins $u$ to $v$.
Indeed, each physical link joins $u$ to $v$, and distilling two $u$--$v$ virtual
edges again yields a $u$--$v$ virtual edge, so the claim follows by induction on the
number of leaves. Consequently any two virtual edges share \emph{both} endpoints, so
\texttt{purify} applies to every pair and no pair is ever rejected as incompatible,
while \texttt{swap}, which requires exactly one shared endpoint, never applies.

\emph{Edge-disjointness.} In a tree, the two subtrees rooted at the children of an
internal node have disjoint leaf sets. Since the leaves are distinct physical links,
the edge sets consumed by the two operands of every internal node are disjoint
automatically, and the disjointness test never rejects a combination.

Hence $\Tree(G_{m})$ is exactly the set of distillation trees over subsets of the
$m$ links, and by Lemma~\ref{lem:doublefact} the trees on a given subset of size $s$
number precisely $T(s)=(2s-3)!!$, so that
$L=\sum_{s=1}^{m}\binom{m}{s}T(s)$. Trees built over distinct leaf sets are
themselves distinct, so the terms of this sum count pairwise disjoint families and
any single term is already a lower bound for $L$. We take $s=m-1$, for which no
further justification is needed: every \emph{proper} subset of the $m$ links remains
in the working list, only the full-$m$-edge trees being diverted (and even those are
constructed). This yields
\begin{equation}
  L \;\ge\; \binom{m}{m-1}\,T(m-1)\;=\;m\,(2m-5)!! ,
\end{equation}
the factor $\binom{m}{m-1}=m$ being the choice of which link to omit and
$T(m-1)=(2m-5)!!$ the number of distillation trees over the retained links.
By Stirling's approximation,
$(2m-5)!! = \dfrac{(2m-4)!}{2^{m-2}(m-2)!} = \exp\!\big(\Theta(m\log m)\big)$,
hence $L = 2^{\Omega(m\log m)}$.

\smallskip
\emph{Upper bound.} For an arbitrary graph the compatibility and disjointness
restrictions can only \emph{remove} trees, so $\Tree(G)$ is contained in the set of
all full binary trees over subsets of $\edges$ and the unconstrained count of
Lemma~\ref{lem:doublefact} over-estimates $L$ in the safe direction. Summing over
all leaf-subsets,
\begin{equation}
  L \;\le\; \sum_{s=1}^{m}\binom{m}{s}\,T(s)
     \;\le\; m\cdot 2^{m}\cdot (2m)^{m}
     \;=\; 2^{\bigO(m\log m)},
\end{equation}
using $\binom{m}{s}\le 2^{m}$ and $(2s-3)!!\le (2m)^{m}$. Therefore
$L = 2^{\Theta(m\log m)} = m^{\Theta(m)}$.

\smallskip
\emph{Running time.} The outer loop traverses the final list of length $L$. For the
$i$-th element, the inner loop performs $i$ disjointness tests with no early
termination, giving $\sum_{i<L} i = \Theta(L^{2})$ bitmask tests, each of constant
cost. The successful combinations number exactly $L-m$ and each is evaluated in
constant time. Hence the total time is $\Theta(L^{2}) = 2^{\Theta(m\log m)}$, assuming $O(1)$ complexity for the fidelity calculations and constraints checking for each pair.
\end{proof}

\paragraph{Instance-Dependent Behavior}

The worst case of \cref{thm:brute} is reached on $G_m$ because the two restrictions
of \cref{def:optree}, compatibility and edge-disjointness, are vacuous there. On
other topologies they bite, and the number of operation trees collapses
accordingly. \Cref{prop:star_path} makes precise the two other points of the range
announced at the start of this section: on a star the constraints admit only trees
of depth one, and on a path only the bracketings of a contiguous segment.

\begin{proposition}[Star and path]
\label{prop:star_path}
Let $S_m$ be the star with $m$ edges and $P_m$ the simple path with $m$ edges.
Then
\begin{equation}
  L(S_m)\;=\;m+\binom{m}{2}\;=\;\Theta(m^{2}),
  \qquad
  L(P_m)\;=\;2^{\Theta(m)} ,
\end{equation}
so the algorithm runs in $\Theta(m^{4})$ time on the star and in $2^{\Theta(m)}$
time on the path.
\end{proposition}

\begin{proof}
\emph{Star.} Write $c$ for the center and $\ell_1,\dots,\ell_m$ for the leaves of the original network graph, so
that every physical edge is $e_i=\{c,\ell_i\}$. Two physical edges share exactly
the center, hence swap into a virtual edge with $\mathrm{ends}=\{\ell_i,\ell_j\}$
and $\mathrm{used}=\{e_i,e_j\}$. They never share both endpoints, so
\texttt{purify} is never applicable at this level. Such a leaf--leaf virtual edge
admits no further combination. Against a physical edge $e_k$ with $k\notin\{i,j\}$
its endpoints $\{\ell_i,\ell_j\}$ and $\{c,\ell_k\}$ are disjoint, so the pair is
incompatible, while $k\in\{i,j\}$ is rejected by the disjointness test. Against
another leaf--leaf virtual edge $\{\ell_k,\ell_l\}$ edge-disjointness forces
$\{i,j\}\cap\{k,l\}=\emptyset$, which again leaves the endpoint sets disjoint.
Hence no tree with three or more leaves exists and $\Tree(S_m)$ consists of the
$m$ physical edges together with the $\binom{m}{2}$ two-leaf swaps.

\emph{Path.} Label the edges $e_1,\dots,e_m$ in order along the path. Every
virtual edge is a contiguous segment $e_i,\dots,e_j$: this holds for the physical
edges, and swapping two edge-disjoint segments that share a node concatenates them
into a segment, while two edge-disjoint segments never share both endpoints, so
\texttt{purify} never applies. A tree over a segment of $s$ edges is therefore
determined by recursively choosing where to cut it, i.e.\ by a binary bracketing
of $s$ ordered factors, of which there are the Catalan number
$C_{s-1}=\frac{1}{s}\binom{2s-2}{s-1}$. The index is shifted because $C_{j}$
counts the bracketings of $j+1$ factors, and a segment of $s$ edges contributes
$s$ leaves. A path of $m$ edges has exactly $m-s+1$ segments of length $s$ (one
per choice of the left index at which the segment starts), and trees over
distinct segments are distinct, so summing over $s$ and using
$2^{\,j-1}\le C_{j}\le 4^{\,j}$ gives
\begin{equation}
  2^{m-2}\;\le\;C_{m-1}\;\le\;L(P_m)\;=\;\sum_{s=1}^{m}(m-s+1)\,C_{s-1}
  \;\le\;m\sum_{s=1}^{m}4^{s-1}\;\le\;m\,4^{m},
\end{equation}
the lower bound being the single term $s=m$. Hence $L(P_m)=2^{\Theta(m)}$.

In both cases the running time is $\Theta(L^{2})$ by the argument at the end of
\cref{thm:brute}.
\end{proof}

The comparison with \cref{lem:doublefact} is the point: on a segment of $s$ edges
the linear order of the path fixes which subtrees may meet, leaving $C_{s-1}<4^{s}$
bracketings out of the $(2s-3)!!=s^{\Theta(s)}$ unordered shapes, whereas on $G_m$
every shape remains available. The star is the degenerate case in which even the second
level of the recursion is closed off.

\section{Complexity of FOLD}
\label{sec:heuristic_complexity}

FOLD (Fidelity-Optimizing Local Detours), the heuristic of \cref{sec:heuristic_algorithm},
proceeds in two phases: it computes the $k$ shortest loopless paths
by Yen's algorithm, then folds each successive path into a running solution,
merging a diverging segment whenever the two parallel sub-paths lie in the
purifiable region and the merge strictly increases the overall metric.
We analyze the two phases separately and add them up in \cref{thm:heur}. The
interesting part is the second phase, whose cost is not governed by the size of
the graph but by its \emph{circuit rank} $r(G)=m-n+c$, with $c$ the number of
connected components: as we show in \cref{lem:rank}, each successful merge
consumes one independent cycle of $G$, so at most $r(G)$ merges can ever occur.
As in \cref{thm:brute}, we assume throughout that a single fidelity evaluation,
$\Fpur$ or $\Fswap$, costs $\bigO(1)$.

\subsection{Phase 1: Yen's \texorpdfstring{$k$}{k}-shortest paths}

The candidate paths are produced by Yen's algorithm, which computes the $k$ shortest
loopless paths in
\begin{equation}
\label{eq:yen}
  T_{\mathrm{Yen}} \;=\; \bigO\!\big(k\,n\,(m+n\log n)\big)
\end{equation}
time on a graph with $n$ nodes and $m$ edges. Our implementation uses a binary heap
rather than a Fibonacci heap for the underlying shortest-path computation, which
replaces \cref{eq:yen} by $\bigO(k\,n\,m\log n)$.

\subsection{Phase 2: How Many Merges Can Succeed}

Phase~2 repeats \emph{improvement steps}: it scans the diverging segments of the
current candidate path against the running solution and applies the first merge
that is legal and profitable, until a full scan finds none
(\cref{alg:path_purification}). Bounding this phase therefore needs two
ingredients, the cost of one step and the number of steps. The second is the
delicate one, since the \textbf{repeat} loop carries no explicit iteration
budget: nothing in the pseudo-code says that merges cannot keep succeeding.
What stops them is a property of the graph rather than of the code, and this is
what we establish here.

\paragraph{Setup} Write $H_t=(V_t,U_t)$ for the sub-multigraph of $G$ spanned by
the physical edges $U_t\subseteq\edges$ that the running solution consumes after
$t$ successful merges. Initially $H_0$ is the first Yen path $p_1$. A successful merge acts
on a pair of diverging sub-paths $(p_1',p_2')\in\mathrm{Div}(p_1,P[i])$ running
between two \emph{anchor} nodes $a,b$ shared by the two paths: it replaces the
segment $p_1'$ of the running solution by the distillation of $p_1'$ with $p_2'$,
so the edges of $p_2'$ join the solution and $U_{t+1}=U_t\cup\mathrm{used}(p_2')$,
a disjoint union.

\paragraph{The idea} The old segment $p_1'$ and the new detour $p_2'$ join the
same two anchors and, by the disjointness test of \cref{alg:path_purification},
share no physical edge, so together they close a cycle of $G$. Every successful
merge closes such a cycle out of edges never used before, so distinct merges close
independent cycles. A graph contains only $r(G)=m-n+c$ independent cycles, hence
at most $r(G)$ merges can ever succeed, no matter how many candidate paths are
folded in, in what order, or with which fidelities. We now make this precise by
tracking the circuit rank of the running solution itself, which starts at $0$,
gains at least one unit per merge, and can never exceed $r(G)$.

\begin{lemma}[Circuit-rank bound on the number of merges]
\label{lem:rank}
The running solution satisfies $r(H_0)=0$, every successful merge satisfies
$r(H_{t+1})\ge r(H_t)+1$, and $r(H_t)\le r(G)$ holds throughout. Consequently an
entire execution of FOLD
performs at most $r(G)=m-n+c$ successful merges.
\end{lemma}

\begin{proof}
The circuit rank of a graph $H$ is the number of its independent cycles, that is,
the dimension of its cycle space, and it is counted for any graph by
\begin{equation}
\label{eq:rank}
  r(H)\;=\;\lvert\text{edges}(H)\rvert-\lvert\text{nodes}(H)\rvert
           +\lvert\text{components}(H)\rvert .
\end{equation}

\emph{Start.} $H_0$ is the shortest path $p_1$, a simple path and hence a tree:
it has one component and one edge fewer than it has nodes, so \cref{eq:rank}
gives $r(H_0)=0$.

\emph{One merge adds at least one independent cycle.} Consider the
$(t{+}1)$-th successful merge, on a diverging segment with anchors $a,b$. Both
anchors already belong to the running solution, so $a,b\in V_t$, and they lie in
the same connected component of $H_t$ because the running solution is a connected
$s$--$t$ structure. The merge adjoins the edge set $P=\mathrm{used}(p_2')$ of the
detour, which is non-empty and, by the disjointness test, contains no edge of
$U_t$. Being a simple $a$--$b$ path, $P$ contributes $\lvert P\rvert$ edges and
$\lvert P\rvert-1$ internal nodes, of which some number $\nu\le\lvert P\rvert-1$
are new to $H_t$ (the others, and the anchors themselves, are already in $V_t$).
Adjoining a path between two nodes of one existing component leaves the component
count unchanged, so each term of \cref{eq:rank} moves by a known amount:
\begin{equation}
  \Delta r \;=\; \underbrace{\lvert P\rvert}_{\Delta\,\text{edges}}
             \;-\; \underbrace{\nu}_{\Delta\,\text{nodes}}
             \;+\; \underbrace{0}_{\Delta\,\text{components}}
           \;\ge\; \lvert P\rvert-\big(\lvert P\rvert-1\big)\;=\;1 .
\end{equation}
The count covers every successful merge, including the degenerate one that leaves
the list of anchors unchanged and only raises the fidelity: it too adjoins a fresh
detour and hence spends a cycle.

\emph{Ceiling.} $H_t$ is an edge-subgraph of $G$, so its cycle space embeds in
that of $G$ and $r(H_t)\le r(G)$ for every $t$.

Combining, after $t$ successful merges $t\le r(H_t)\le r(G)$, which is the claim.
\end{proof}

\begin{remark}[Distillation is a phenomenon of the cycle space]
\label{rem:cyclespace}
\Cref{lem:rank} says that a merge \emph{spends} an independent cycle: the one
formed by the old $a$--$b$ segment of the running solution together with the fresh
parallel detour. Once the solution has absorbed the whole cycle space of $G$, no
edge-disjoint alternative route remains between any two of its nodes and folding
halts. The degenerate regime is instructive: a forest has $r=0$ and admits no
merge at all, which is exactly right, since between any two nodes a tree offers a
unique path and there is nothing to distill against. The bound is thus
topology-aware rather than size-aware. Near-tree networks admit almost no merges
however many candidate paths are examined, while dense or heavily parallel
networks, for which $r=\Theta(m)$, admit up to $\Theta(m)$ of them. It is the same
quantity that drives the super-exponential growth of the exhaustive scheme in
\cref{thm:brute}, where the witness $G_m$ has $r(G_m)=m-1$.
\end{remark}

\subsection{Cost of One Improvement Step and Total Running Time}

A single improvement step recomputes the fidelity of the running solution with
$\bigO(n)$ primitive evaluations, collects the identifiers of the edges it
consumes in $\bigO(m)$, and scans at most $n$ diverging segments, one per pair of
consecutive anchors. Each segment costs $\bigO(n)$, since the
purifiable-region and improvement tests sweep sub-paths of at most $n$ hops a
bounded number of times, at $\bigO(1)$ per hop, and splicing the segment in
requires locating each of its two anchors in both the running solution and the
candidate path, which is four list-index lookups, also at $\bigO(n)$. The index lookups dominate a short segment irrespective of its
length, and a position map would remove them, but we analyze the algorithm as
written. One improvement step therefore costs $\bigO(n^{2}+m)$.

By \cref{lem:rank} at most $r(G)$ improvement steps ever succeed. Each of the
$k-1$ folding passes additionally ends with a single scan that finds no
improvement and terminates the \textbf{repeat} loop, so the total number of
improvement steps over an execution is at most $r(G)+k$.

\begin{theorem}[Complexity of FOLD]
\label{thm:heur}
FOLD (\cref{alg:path_purification}) runs in
\begin{equation}
  T_{\mathrm{heur}}
  \;=\;
  \underbrace{\bigO\!\big(k\,n\,m\log n\big)}_{\text{Yen}}
  \;+\;
  \underbrace{\bigO\!\big((r+k)\,(n^{2}+m)\big)}_{\text{folding}},
  \qquad r=m-n+c .
\end{equation}
For a fixed path budget $k=\bigO(1)$ this is the topology-aware bound
\begin{equation}
  \boxed{\;T_{\mathrm{heur}} \;=\; \bigO\!\big(r\,(n^{2}+m) \;+\; n\,m\log n\big)\;},
\end{equation}
which is polynomial in the size of the network.
\end{theorem}

\begin{proof}
Add the Phase-1 bound of \cref{eq:yen}, in the binary-heap form
$\bigO(k\,n\,m\log n)$ used by our implementation, to the Phase-2 cost, which is
the number of improvement steps, at most $r+k$, times the per-step cost
$\bigO(n^{2}+m)$. The final reduction of the running solution to a single $s$--$t$
virtual edge costs a further $\bigO(n)$ and is absorbed. Setting $k=\bigO(1)$
leaves the two displayed terms. With a Fibonacci-heap shortest-path routine the
Yen term would instead read $\bigO(nm+n^{2}\log n)$, by \cref{eq:yen}.
\end{proof}

\begin{corollary}[Worst case]
\label{cor:dense}
On a simple graph, where $m=\bigO(n^{2})$, and with a fixed path budget,
\cref{thm:heur} gives $T_{\mathrm{heur}}=\bigO(m\,n^{2})$, since $r\le m$ and
$n\,m\log n=\bigO(m\,n^{2})$. The bound is attained up to constants when
$r=\Theta(m)$, i.e.\ on dense or heavily parallel networks. Conversely, on
near-tree networks, where $r=\bigO(1)$, the cost is Yen-dominated,
$T_{\mathrm{heur}}=\bigO(n\,m\log n)$, with no $m\,n^{2}$ term at all.
\end{corollary}

This is the polynomial complexity quoted in \cref{sec:heuristic_algorithm}, and it
stands in contrast to the $2^{\Theta(m\log m)}$ of \cref{thm:brute}: FOLD
examines at most one merge per independent cycle along at most $k$ candidate paths,
whereas the exhaustive scheme traverses the entire lattice of operation trees.

\subsection{Cost of the Path Budget}

The path budget $k_{\max}$ is FOLD's main free parameter, so we close by
tracing what it costs. \cref{thm:heur} already locates the whole dependence. The
budget enters the additive Yen term and the at most $k_{\max}$ terminating scans,
and nowhere else. It never reaches the folding work, which \cref{lem:rank} caps at
$r(G)$ successful merges no matter how many candidates are folded in. Raising the
budget therefore supplies more candidate paths without enlarging the expensive
phase, and the running time grows gently rather than in proportion to $k_{\max}$.
Timing FOLD across a range of budgets, on the Erd\H{o}s--R\'enyi ensemble
of \cref{fig:heuristic_benchmark}, bears this out (\cref{fig:kmax_runtime}). A
generous budget is cheap, which is what lets FOLD search many routes for
little more than the price of one.

\begin{figure}[H]
  \centering
  \includegraphics[width=\columnwidth]{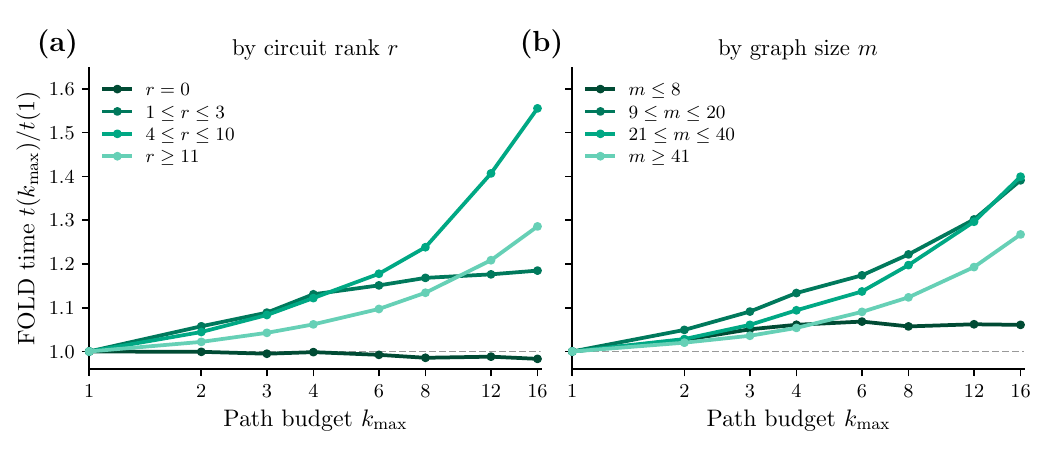}
  \caption{\textbf{The path budget is cheap.} Each Erd\H{o}s--R\'enyi instance of
  \cref{fig:heuristic_benchmark} is solved by FOLD at every budget
  $k_{\max}$, and the curves show the median across instances of the per-graph
  ratio $t(k_{\max})/t(1)$, the running time relative to a single shortest path.
  \textbf{(a)} Instances grouped by circuit rank $r=m-n+c$, the quantity that
  bounds the number of merges. \textbf{(b)} The same runs grouped by the number of
  physical edges $m$. Increasing the budget sixteen-fold changes the median run by
  only about a third, and on trees ($r=0$), where no merge is possible, the time
  is flat. The budget feeds the $k$-shortest-paths search but not the
  rank-bounded merging phase (\cref{thm:heur}), so its cost stays mild. Absolute
  times are milliseconds, with a median of $30$ ms at $k_{\max}=1$.}
  \label{fig:kmax_runtime}
\end{figure}

\section{In What Order to Fold: Distillation Schedules}
\label{sec:distillation_order}

Distillation is defined for two states but is applied to several, so a schedule must be chosen. The three natural ones are drawn in \cref{fig:distillation_strategy}: \emph{parallel} distills all disjoint pairs in a single round, \emph{sequential} keeps one working state and purifies it against each remaining state in turn, and \emph{tournament} combines the states in balanced rounds. The question is not new. \citet{vanmeter2009} analyze exactly this choice for long-line repeaters, contrasting entanglement pumping with symmetric schedules and introducing \emph{banded} purification, and \citet{goodenough2021} later carry both ideas into their protocol search, with scheduling work on linear networks continuing to optimize the ordering directly~\citep{koutsopoulos2024,wang2023scheduling}.

Those results are derived for a repeater chain fed by a continuous flow of fresh pairs. \cref{fig:distillation_strategy}(d) shows that the same ranking is reproduced in the one-shot model, where nothing is replenished and a schedule must work with the states already on hand: sweeping our simulator across the three schedules puts sequential pumping ahead by a wide margin, while the rule used to choose which edges to distill has a negligible effect. We therefore adopt the pumping order rather than re-derive it.

\providecommand{\dpair}[4]{
  \draw[qrDark, line width=0.6pt] (#2,#3) -- ++(0.52,0);
  \fill[qrDark] (#2,#3) circle (1.25pt);
  \fill[qrDark] (#2+0.52,#3) circle (1.25pt);
  \coordinate (#1l) at (#2,#3);
  \coordinate (#1r) at (#2+0.52,#3);
  \node[font=\scriptsize, text=qrDark, anchor=south, inner sep=1.5pt]
       at (#2+0.26,#3) {#4};
}

\begin{figure}[H]
  \centering
  \begin{subfigure}[t]{0.20\columnwidth}
    \centering
    \begin{tikzpicture}[scale=1.2,
                          dist/.style={draw=qrMid, line width=0.6pt, dashed,
                                       -{Stealth[length=1.4mm]}}]
      \dpair{a}{0}{0}{$e_1$}
      \dpair{b}{0}{-0.6}{$e_2$}
      \dpair{c}{0}{-1.2}{$e_3$}
      \dpair{d}{0}{-1.8}{$e_4$}
      \dpair{u}{1.15}{-0.3}{}
      \dpair{v}{1.15}{-1.5}{}
      \draw[dist] (ar) -- (ul);
      \draw[dist] (br) -- (ul);
      \draw[dist] (cr) -- (vl);
      \draw[dist] (dr) -- (vl);
    \end{tikzpicture}
    \caption{Parallel}
    \label{fig:dist_parallel}
  \end{subfigure}
  \hfill
  \begin{subfigure}[t]{0.45\columnwidth}
    \centering
    \begin{tikzpicture}[scale=1.2,
                          dist/.style={draw=qrMid, line width=0.6pt, dashed,
                                       -{Stealth[length=1.4mm]}}]
      \dpair{a}{0}{0}{$e_1$}
      \dpair{b}{0}{-0.6}{$e_2$}
      \dpair{c}{0}{-1.2}{$e_3$}
      \dpair{d}{0}{-1.8}{$e_4$}
      \dpair{u}{1.15}{-0.3}{}
      \dpair{v}{2.30}{-0.75}{}
      \dpair{w}{3.45}{-1.275}{}
      \draw[dist] (ar) -- (ul);
      \draw[dist] (br) -- (ul);
      \draw[dist] (ur) -- (vl);
      \draw[dist] (cr) -- (vl);
      \draw[dist] (vr) -- (wl);
      \draw[dist] (dr) -- (wl);
    \end{tikzpicture}
    \caption{Sequential}
    \label{fig:dist_sequential}
  \end{subfigure}
  \hfill
  \begin{subfigure}[t]{0.32\columnwidth}
    \centering
    \begin{tikzpicture}[scale=1.2,
                          dist/.style={draw=qrMid, line width=0.6pt, dashed,
                                       -{Stealth[length=1.4mm]}}]
      \dpair{a}{0}{0}{$e_1$}
      \dpair{b}{0}{-0.6}{$e_2$}
      \dpair{c}{0}{-1.2}{$e_3$}
      \dpair{d}{0}{-1.8}{$e_4$}
      \dpair{u}{1.15}{-0.3}{}
      \dpair{v}{1.15}{-1.5}{}
      \dpair{w}{2.30}{-0.9}{}
      \draw[dist] (ar) -- (ul);
      \draw[dist] (br) -- (ul);
      \draw[dist] (cr) -- (vl);
      \draw[dist] (dr) -- (vl);
      \draw[dist] (ur) -- (wl);
      \draw[dist] (vr) -- (wl);
    \end{tikzpicture}
    \caption{Tournament}
    \label{fig:dist_tournament}
  \end{subfigure}

  \vspace{1em}
  \begin{subfigure}[t]{0.85\columnwidth}
    \centering
    \includegraphics[width=\linewidth]{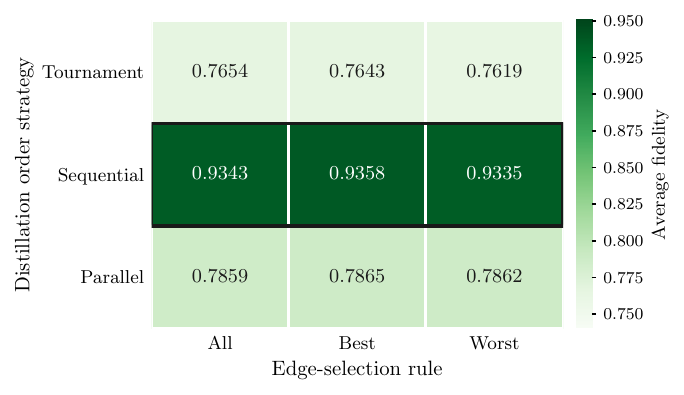}
    \caption{Average delivered fidelity}
    \label{fig:dist_matrix}
  \end{subfigure}

  \caption{\textbf{The three schedules for distilling several parallel states,
  and how they compare.} Each entangled state is drawn as two endpoints joined by a line,
  and each dashed arrow is one two-to-one distillation.
  \textbf{(a)} \emph{Parallel} distills all disjoint pairs in a single
  simultaneous round, leaving several states rather than one.
  \textbf{(b)} \emph{Sequential} keeps one working state and purifies it against
  each remaining state in turn, the entanglement-pumping order of
  \citet{vanmeter2009}.
  \textbf{(c)} \emph{Tournament} combines the states in balanced rounds, pairing
  those of similar fidelity.
  \textbf{(d)} Average fidelity of the distilled state for each schedule (rows)
  and edge-selection rule (columns), obtained from the measured negativity via
  the Werner relation $F=N+\tfrac12$ with failed distillations counted as
  $F=\tfrac12$, so that the score charges for the probability of success. On
  this measure the sequential schedule delivers the highest average fidelity in
  our one-shot model, reproducing the ranking \citet{vanmeter2009}
  establish for repeater chains supplied with fresh pairs, whereas the
  edge-selection rule has a negligible effect.}
  \label{fig:distillation_strategy}
\end{figure}
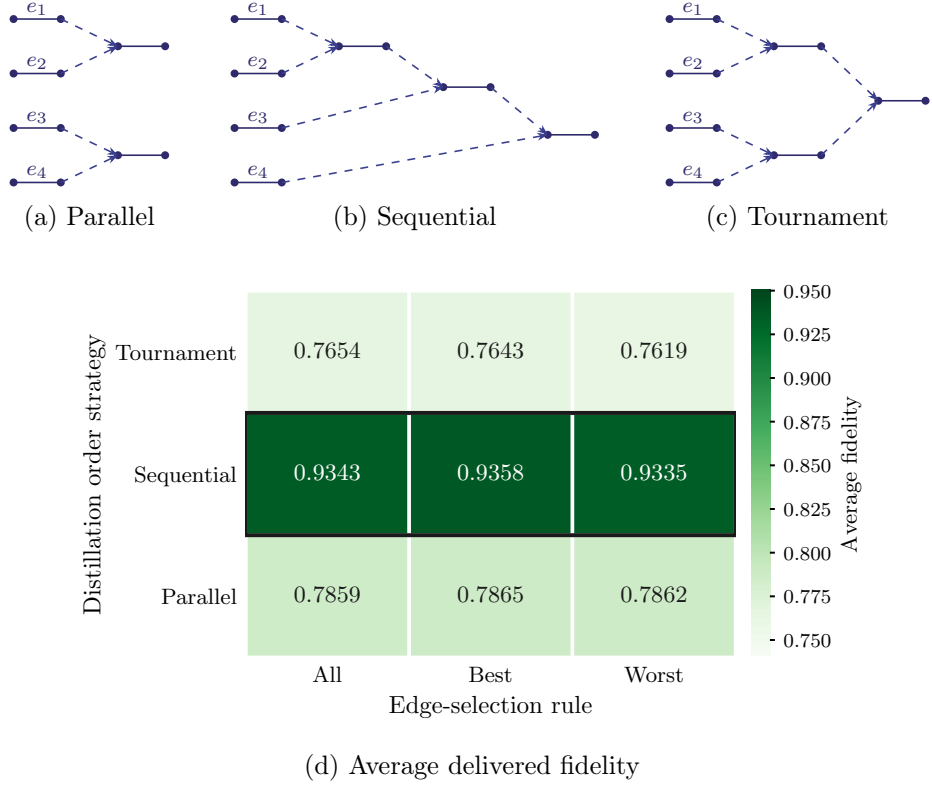

One qualification is worth stating, because it shapes FOLD. The sweep scores each schedule by the fidelity actually delivered, charging a failed distillation as $F=\tfrac12$, so it rewards a schedule that reaches a usable state reliably. Under the fidelity-only objective of \cref{sec:model}, which sets success probability aside, the three schedules lie much closer together and the ranking can invert when the states being combined are \emph{equal}: distilling four identical $A$--$B$ links in balanced rounds gives $F=0.919$, against $F=0.911$ for pumping them in one at a time. The two statements are not in conflict, pumping being the robust choice and a balanced order the fidelity-optimal one for equal inputs, but the distinction matters here because an algorithm that maintains a single running solution and folds candidates into it is, structurally, a pumping schedule.
\SupplementalReferences

\end{document}